\documentclass[12pt]{article}%
\usepackage{amsmath,amssymb,eucal}
\usepackage{amsfonts}
\usepackage{mathrsfs}
\usepackage{slashed}
\numberwithin{equation}{section} 
\def\beq{\begin{equation}}
\def\eeq{\end{equation}}
\def\bE{ {{\mathbb{E}}}}
\def\bC{ {{\mathbb{C}}}}
\def\bR{ {{\mathbb{R}}}}

\def\Tr{ {{\rm{Tr}}} }

\def\fd{ {\mathfrak{d}} }
\newcommand{\dB}{{\mathbb{B}}}

\newtheorem{defn}{{\bf Definition}}[section]
\newtheorem{thm}[defn]{{\bf Theorem}}
\newtheorem{cor}[defn]{{\bf Corollary}}
\newtheorem{lem}[defn]{{\bf Lemma}}
\newtheorem{prop}[defn]{{\bf Proposition}}
\newtheorem{rem}[defn]{{\bf Remark}}

\newtheorem{notation}[defn]{Notation}
\newenvironment{proof}[1][Proof]{\textbf{#1.} }{\hfill \rule{0.5em}{0.5em}}

\begin{document}

\title{Wilson Area Law formula on $\mathbb{R}^4$}
\author{Adrian P. C. Lim \\
Email: ppcube@gmail.com
}

\date{}

\maketitle

\begin{abstract}
Let $\mathfrak{g}$ be the Lie Algebra of a compact semi-simple gauge group. For a $\mathfrak{g}$-valued 1-form $A$, consider the Yang-Mills action \begin{equation} S_{{\rm YM}}(A) = \int_{\mathbb{R}^4}  \left|dA + A \wedge A \right|^2\ d\omega, \nonumber
\end{equation}
using the Euclidean metric on $T\mathbb{R}^4$. We want to make sense of the following path integral,
\begin{equation}
{\rm Tr}\ \int_{A \in \mathcal{A}_{\mathbb{R}^4, \mathfrak{g}} /\mathcal{G}}  \exp \left[  c\int_{S} dA\right] e^{-\frac{1}{2}S_{{\rm YM}}(A)}\ DA, \nonumber
\end{equation}
whereby $DA$ is some Lebesgue type of measure on the space of $\mathfrak{g}$-valued 1-forms, modulo gauge transformations $\mathcal{A}_{\mathbb{R}^4, \mathfrak{g}} /\mathcal{G}$. Here, $S$ is some compact flat rectangular surface.

Using an Abstract Wiener space, we can define a Yang-Mills path integral rigorously, for a compact semi-simple gauge group. Subsequently, we will then derive the Wilson area law formula from the definition, using renormalization techniques and asymptotic freedom.

One of the most important applications of the Area Law formula will be to explain why the potential measured between a quark and antiquark is a linear function of its distance.
\end{abstract}

\hspace{.35cm}{\small {\bf MSC} 2020: 81T13, 81T08} \\
\indent \hspace{.35cm}{\small {\bf Keywords}: Yang-Mills measure, axial gauge fixing, area law, renormalization, \\
\indent \hspace{2.5cm} asymptotic freedom}




\section{Preliminaries}\label{s.pre}

This is a sequel to the article \cite{YMLim01}, which constructs a Yang-Mills measure using an abelian gauge group, in $\bR^4$. We will now consider a compact semi-simple gauge group. Balaban from 1984 to 1989, published a series of papers, defining a Yang-Mills path integral using finite dimensional lattice gauge. The references for these papers can be found in \cite{YMLim01}.

Instead of using lattice gauge approximation, we will use the Segal Bargmann Transform, to construct a path integral on the Hilbert space containing $\mathfrak{g}$-valued holomorphic functions. Renormalization techniques used in the abelian gauge group case will also be applied in the non-abelian case. But the most notable difference is that the Wilson Area Law formula does not hold in the abelian case; however by assuming asymptotic freedom, we will show that it holds in the non-abelian case.

In Balaban's series of papers, asymptotic freedom was not used. Rivasseau and his co-authors in \cite{rivasseau01}, applied asymptotic freedom to construct a Yang-
Mills integral in 4-dimensional space, but the analysis is presented specifically for $SU(N)$ gauge group. Both Rivasseau and Balaban considered infra-red cutoff in their analysis, and proved ultraviolet stability. But they did not show that the Area Law formula holds. We will also refer the reader to \cite{2018arXiv180301950C} for a comprehensive review on work done in non-abelian Yang-Mills path integrals.

Consider a 4-manifold $M$ and a principal bundle $P$ over $M$, with structure group $G$. We assume that $G$ is compact and semi-simple. Without loss of generality we will assume that $G$ is a Lie subgroup of $U(\bar{N})$, $\bar{N} \in \mathbb{N}$ and $\pi: P \rightarrow M$ is a trivial bundle. We will identify the Lie Algebra $\mathfrak{g}$ of $G$ with a Lie subalgebra of the Lie Algebra $\mathfrak{u}(\bar{N})$ of $U(\bar{N})$ throughout this article. Suppose we write $\Tr \equiv \Tr_{{\rm Mat}(\bar{N}, \mathbb{C})}$ as matrix trace.
Then we can define a positive, non-degenerate bilinear form by \beq \langle A, B \rangle = -\Tr_{{\rm Mat}(\bar{N}, \mathbb{C})}[AB] \label{e.i.2} \eeq for $A,B \in \mathfrak{g}$. Let $\{E^\alpha\}_{\alpha=1}^N$ be an orthonormal basis in $\mathfrak{g}$, which will be fixed throughout this article. Refer to \cite{MR499562}, \cite{MR2062813} and \cite{MR2279709}.

The vector space of all smooth $\mathfrak{g}$-valued 1-forms on the manifold $M$ will be denoted by $\mathcal{A}_{M, \mathfrak{g}}$. Denote the group of all smooth $G$-valued mappings on $M$ by $\mathcal{G}$, called the gauge group. The gauge group induces a gauge transformation on $\mathcal{A}_{M, \mathfrak{g}}$,  $\mathcal{A}_{M, \mathfrak{g}} \times \mathcal{G} \rightarrow \mathcal{A}_{M, \mathfrak{g}}$ given by \beq A \cdot \Omega := A^{\Omega} = \Omega^{-1}d\Omega + \Omega^{-1}A\Omega \nonumber \eeq for $A \in \mathcal{A}_{M, \mathfrak{g}}$, $\Omega \in \mathcal{G}$. The orbit of an element $A \in \mathcal{A}_{M, \mathfrak{g}}$ under this operation will be denoted by $[A]$ and the set of all orbits by $\mathcal{A}_{M, \mathfrak{g}}/\mathcal{G}$.

Let $\Lambda^q(T^\ast M)$ be the $q$-th exterior power of the cotangent bundle over $M$. Fix a Riemannian metric $g$ on $M$ and this in turn defines an inner product $\langle \cdot, \cdot \rangle_q$ on $\Lambda^q(T^\ast M)$, for which we can define a volume form $d\omega$ on $M$. This allows us to define a Hodge star operator $\ast: \Lambda^k(T^\ast M) \rightarrow \Lambda^{4-k}(T^\ast M)$ such that for $u, v \in \Lambda^k(T^\ast M)$, we have \beq u \wedge \ast v = \langle u, v \rangle_k\ d\omega. \label{e.x.7} \eeq  An inner product on the set of smooth sections $\Gamma(\Lambda^k(T^\ast M))$ is then defined as \beq \langle u, v \rangle = \int_M u \wedge \ast v = \int_M \langle u, v \rangle_k\ d\omega. \label{e.x.7a} \eeq See \cite{MR1312606}.

Given $u \otimes E \in \Lambda^q(T^\ast M) \otimes (M  \times \mathfrak{g}\rightarrow M)$, we write \beq |u \otimes E|^2\ d\omega = -\Tr\left[ E \cdot E \right]\ u \wedge \ast u = -\Tr\left[ E \cdot E \right]\langle u, u \rangle_q\ d\omega, \nonumber \eeq computed at each fiber of the tensor bundle. Hence for $A \in \mathcal{A}_{M, \mathfrak{g}}$, the Yang-Mills action is given by \beq S_{{\rm YM}}(A) = \int_{M}  \left|dA + A \wedge A \right|^2\ d\omega. \label{e.c.1} \eeq Note that this action is invariant under gauge transformations.

\begin{notation}
From now on, we only consider $M = \bR^4 \equiv \bR \times \bR^3$, $\bR$ is the time-axis and $\bR^3$ is spatial space. Take the principal bundle $P$ over $\bR^4$ to be the trivial bundle. Let $\Lambda^p(\bR^n)$ be the $p$-th exterior power of the dual to the vector space $\bR^n$, $n=3,4$.

On $\bR^4$, fix the coordinate axes with corresponding coordinates $(x^0, x^1, x^2, x^3)$, $x^0$ is the time coordinate. We will also choose the standard Riemannian metric (Euclidean metric) on $T\mathbb{R}^4$. Now let $T^\ast\bR^4 \rightarrow \bR^4$ denote the trivial cotangent bundle over $\bR^4$, i.e. $T^\ast \bR^4 \cong \bR^4 \times \Lambda^1(\bR^4)$. Note that $\Lambda^1(\bR^4)$ is spanned by $\{dx^0, dx^1, dx^2, dx^3\}$ and $\Lambda^1(\bR^3)$ denote the subspace in $\Lambda^1(\bR^4)$, spanned by $\{dx^1, dx^2, dx^3\}$.
\end{notation}

Using axial gauge fixing, every $A \in \mathcal{A}_{\bR^4, \mathfrak{g}} /\mathcal{G}$ can be gauge transformed into a $\mathfrak{g}$-valued 1-form, of the form $A = \sum_{\alpha=1}^N \sum_{j=1}^3a_{j, \alpha} \otimes  dx^j \otimes E^\alpha $,  subject to the conditions \beq a_{1,\alpha}(0, x^1, 0, 0) = 0,\ a_{2, \alpha}(0, x^1, x^2, 0)=0,\ a_{3,\alpha}(0, x^1, x^2, x^3) = 0. \nonumber \eeq Hence it suffices to consider the Yang-Mills integral in Expression (\ref{e.ym.1}) to be over the space of $\mathfrak{g}$-valued 1-forms of the form $A = \sum_{\alpha=1}^N \sum_{j=1}^3a_{j, \alpha} \otimes  dx^j \otimes E^\alpha $, whereby $a_{j,\alpha}: \bR^4 \rightarrow \bR$ is smooth.

Now, \beq dA = \sum_{\alpha=1}^N\sum_{j=1}^3  a_{0:j, \alpha} \otimes dx^0 \wedge dx^j \otimes E^\alpha + \sum_{\alpha=1}^N\sum_{1\leq i<j \leq 3}a_{i;j, \alpha} \otimes  dx^i \wedge dx^j \otimes E^\alpha, \label{e.a.3} \eeq for \beq a_{i;j,\alpha} := \partial_i a_{j, \alpha} - \partial_j a_{i, \alpha},\quad a_{0:j, \alpha} := \partial_0 a_{j, \alpha}, \quad \partial_\mu \equiv \frac{\partial}{\partial x^\mu}. \label{e.d.1} \eeq

Its curvature will hence be given by
\begin{align*}
dA + A \wedge A
=& \sum_{\alpha=1}^N\sum_{j=1}^3  a_{0:j, \alpha} \otimes dx^0 \wedge dx^j \otimes E^\alpha + \sum_{\alpha=1}^N\sum_{1\leq i<j \leq 3}a_{i;j, \alpha} \otimes  dx^i \wedge dx^j \otimes E^\alpha   \\
& + \sum_{1\leq \alpha,\beta \leq N}\sum_{1 \leq i<j\leq 3}a_{i,\alpha}a_{j,\beta} \otimes dx^i \wedge  dx^j \otimes [E^\alpha, E^\beta].
\end{align*}
See \cite{MR1031992}. Note that $[\cdot, \cdot]$ is the Lie Bracket.

\begin{notation}\label{n.n.4}
Note that $\Lambda^2(T^\ast \bR^4) \cong \bR^4 \times \Lambda^2(\bR^4)$. Using the standard coordinates $\{x^0, x^1, x^2, x^3\}$, we fix an orthonormal basis $\{dx^1 \wedge dx^2, dx^3 \wedge dx^1, dx^1 \wedge dx^2, dx^0 \wedge dx^1, dx^0 \wedge dx^2, dx^0 \wedge dx^3\}$ in $\Lambda^2 (\bR^4)$, using the standard metric on $T\bR^4$. The corresponding volume form is given by $d\omega = dx^0 \wedge dx^1 \wedge dx^2 \wedge dx^3$.

From the Hodge star operator and the above volume form, we can define an inner product on the set of smooth sections $\Gamma(\Lambda^2( T^\ast \bR^4))$ as in Equation (\ref{e.x.7a}).

\end{notation}

\begin{defn}
Define \beq c_{\gamma}^{\alpha \beta} = -\Tr\left[E^\gamma [E^\alpha, E^\beta] \right]. \nonumber \eeq They are also referred to as the structure constants in physics.
\end{defn}

\begin{rem}
Note that $c_\gamma^{\alpha \beta} \neq 0$, only if $\gamma$, $\alpha$ and $\beta$ are all distinct. See \cite{peskin1995introduction}.
\end{rem}

Then,
\begin{align}
dA &+ A \wedge A \nonumber \\
=& \sum_{\gamma=1}^N\bigg[ \sum_{1\leq i<j \leq 3}a_{i;j, \gamma}\otimes dx^i \wedge dx^j +\sum_{\alpha,\beta}\sum_{1\leq i<j\leq 3}a_{i,\alpha}a_{j,\beta} c_\gamma^{\alpha \beta} \otimes dx^i \wedge  dx^j \nonumber \\
&\hspace{8cm}+ \sum_{j=1}^3  a_{0:j, \gamma}\otimes  dx^0 \wedge dx^j \bigg]\otimes E^\gamma. \nonumber
\end{align}
Thus,
\begin{align}
\int_{\bR^4}|dA + A \wedge A|^2\ d\omega =& \sum_{\gamma=1}^N\sum_{1\leq i<j \leq 3}\int_{\bR^4}\bigg[ a_{i;j, \gamma}^2 + \sum_{\alpha, \beta, \hat{\alpha}, \hat{\beta}}a_{i,\alpha}a_{j,\beta}a_{i,\hat{\alpha}}a_{j,\hat{\beta}}c_{\gamma}^{\alpha\beta}c_\gamma^{\hat{\alpha}\hat{\beta}} \nonumber \\
&+2\sum_{\alpha, \beta=1}^Na_{i;j,\gamma}a_{i,\alpha}a_{j,\beta}c_\gamma^{\alpha\beta} \bigg]\ d\omega + \sum_{\gamma=1}^N\sum_{1\leq j\leq 3}\int_{\bR^4} a_{0:j,\gamma}^2\ d\omega.
\label{e.ym.2}
\end{align}

\begin{defn}\label{d.rat.1}
Suppose $e^0$ is a directional vector in the time direction, with length $|e^0| = T$, and $a \in \bR^3$ is any directional vector, with length $|a|$. Let $\sigma: I^2 \equiv [0,1]^2 \rightarrow \bR^4$ be some parametrization of a compact rectangular surface $R[a,T]$, spanned by $e^0$ and $a$, of dimensions height T and length $|a|$.

Explicitly, for some $\vec{c} \in \bR^4$, we can choose $\sigma(s,t) = \vec{c} + s e^0 + t a \in \bR^4$, $0\leq s, t \leq 1$. For any $\delta \geq 0$, write $I_\delta = [-\delta, 1 + \delta]$ and $I \equiv I_0$. We can extend the parametrization $\sigma$ to be defined on $I_\delta^2 \equiv I_\delta \times I_\delta$. We will write $R_\delta[a,T]$ to be the image of $I_\delta^2$ under $\sigma$, which is a compact rectangular surface containing $R[a,T] \equiv R_0[a,T]$. Note that the dimensions of $R_\delta[a,T]$ is $|a|(1 + 2\delta)$ by $T(1 + 2\delta)$.
\end{defn}

On the surface $R[a,T]$, we see that for $A = \sum_{\alpha=1}^N \sum_{j=1}^3 A_{j,\alpha}\otimes dx^j \otimes E^\alpha$, \beq \int_{R[a,T]} dA + A \wedge A = \int_{R[a,T]} dA. \nonumber \eeq For a coupling constant $c$, we made sense of \beq \frac{1}{Z}\int_{A \in \mathcal{A}_{\bR^4, \mathfrak{g}} /\mathcal{G}}\exp \left[  c\int_{R[a,T]} dA \right] e^{-\frac{1}{2}\int_{\bR^4}|dA|^2\ d\omega} D[A], \nonumber \eeq in \cite{YMLim01}, for an abelian gauge group.

Let $\rho: \mathfrak{g} \rightarrow {\rm End}(\bC^{\tilde{N}})$ be an irreducible representation of $\mathfrak{g}$. Throughout this article, we assume that $\rho(\mathfrak{g})$ consists of skew-Hermitian matrices. Write $A^\rho := \sum_{\alpha=1}^N\sum_{i=1}^3 A_{i,\alpha}\otimes dx^i \otimes \rho(E^\alpha)$, which is $\rho(\mathfrak{g})$-valued.

Motivated by the abelian case, we want to make sense of the following path integral,
\beq \frac{1}{Z}\Tr\int_{A \in \mathcal{A}_{\bR^4, \mathfrak{g}} /\mathcal{G}} \exp \left[  c\int_{R[a,T]} dA^\rho \right] e^{-\frac{1}{2}S_{{\rm YM}}(A)}\ D[A], \label{e.ym.1} \eeq whereby $D[A]$ is some non-existent Lebesgue type of measure, \beq Z = \int_{A \in \mathcal{A}_{\bR^4, \mathfrak{g}} /\mathcal{G}} e^{-\frac{1}{2}S_{{\rm YM}}(A)}\ D[A], \nonumber \eeq and $c$ is a coupling constant. Note that $dA^\rho$ is given by Equation (\ref{e.a.3}), except that we replace $E^\alpha$ in the said equation with $\rho(E^\alpha)$.

\section{Non-abelian Yang-Mills measure}\label{s.naymm}

\begin{notation}\label{n.n.5}
For $x \in \bR^4$, we let $\phi_\kappa(x) = \kappa^4 e^{-\kappa^2 |x|^2/2}/(2\pi)^2$, which is a Gaussian function with variance $1/\kappa^2$. Define a function $\psi_z \equiv \psi(z)$, where $\psi(z) = \frac{1}{\sqrt{2\pi}}e^{-\sum_{i=1}^4|z_i|^2/2}$ and $z \equiv (z_0, z_1, z_2, z_3) \in \bC^4$.
\end{notation}

\begin{rem}\label{r.aym.1}
In \cite{YMLim01}, we explained that the Gaussian function $\phi_\kappa$ will approximate the Dirac Delta function, which is a generalized function. The larger the value of $\kappa$, the better is this approximation. The variance $1/\kappa^2$, denotes how well we can resolve a point in $\bR^4$. In all our Yang-Mills path integrals to be define later, it will depend on the parameter $\kappa$.

A finite dimensional gauge lattice with lattice spacing $\epsilon$ is commonly used to define a finite dimensional path integral, of which one has to take the limit as $\epsilon$ goes down to zero, to define an infinite dimensional path integral. One should think of $1/\kappa$ as being synonymous with the lattice spacing $\epsilon$. We refer the reader to \cite{YMLim01} for a brief account on the use of lattice gauge theory, for the purpose of defining a path integral.
\end{rem}

\subsection{Wiener measure}

In \cite{YMLim01}, we considered the (real) vector space denoted by $\mathcal{S}_\kappa(\bR^4)$, $\kappa > 0$. A vector in the said space is written in the form $p_i \cdot \sqrt{\phi_\kappa} $, whereby $p_i$ is a polynomial on $\bR^4$. We used the $L^2$ inner product on this space.

The Segal Bargmann Transform $\Psi_\kappa$ maps it into the (real) vector space of polynomials, over in $\bC^4$, denoted as $H^2(\bC^4)$. On $\bC^4$, we will write \beq z = (z_0, z_1, z_2, z_3) \in \bC^4\quad {\rm and}\quad d\lambda_4 = \frac{1}{\pi^4}e^{-\sum_{i=0}^3|z_i|^2} \prod_{i=0}^3 dx^i dp^i,\nonumber \eeq be a Gaussian measure, whereby $z_i \equiv x^i + \sqrt{-1}p^i$.

The real inner product on $H^2(\bC^4)$ is given by \beq \langle f, g \rangle := \int_{\bC^4} f(z) \bar{g}(z)\ d\lambda_4, \nonumber \eeq $\bar{g}(z)$ means the complex conjugate of $g(z)$. Let the Hilbert space $\mathcal{H}^2(\bC^4)$ be the completion of $H^2(\bC^4)$ and $\mathcal{H}^2(\bC^4) \otimes_{\bR} \bC = \mathcal{H}^2(\bC^4)_\bC$, the complexification of the Hilbert space. The inner product will now be extended to be sesquilinear.

We also call $\{\Psi_\kappa: \kappa>0\}$ a renormalization flow between \beq \left\{(\mathcal{S}_\kappa(\bR^4),\langle \cdot, \cdot \rangle):\kappa>0 \right\}\quad {\rm and} \quad \left\{(H^2(\bC^4),  \langle \cdot, \cdot \rangle) \right\}. \nonumber \eeq The Segal Bargmann Transform is an isometry between the two inner product spaces. We extend this isometry to be between $\mathcal{S}_\kappa(\bR^4) \otimes \Lambda^2(\bR^4)$ and $H^2(\bC^4) \otimes \Lambda^2(\bR^4)$, using the tensor inner product defined on the respective spaces.

We explained in \cite{YMLim01}, by defining a measurable norm $|\cdot|$, we complete $\mathcal{H}^2(\bC^4) $ into a Banach space $B(\bC^4)$, and form an Abstract Wiener space $(i, \mathcal{H}^2(\bC^4), B(\bC^4))$, whereby $i$ is an inclusion map and $B(\bC^4)$ is equipped with a Wiener measure. This Banach space consists of holomorphic functions over $\bC^4$. For the convenience of the reader, we reproduced the definition of $|\cdot|$ in Definition \ref{d.s.1}.

\subsection{Tensor products}

Suppose we have an Abstract Wiener space, $(i, H, B)$, with a Hilbert space $H \subset B$ which is dense and let the dual space $B^\ast$ be identified with a dense subspace in $H$. By abuse of notation, we will write $B^\ast \subset H$ to denote this identification. Denote the inner product on $H$ with $\langle \cdot, \cdot \rangle$.

We will identify $f \in B^\ast$ with $(\cdot, f)_\sharp$ as a random variable on $(B, \mu)$, $\mu$ is Wiener measure on $B$ with variance $\lambda$. Thus for $x \in B$, $(x, f)_\sharp \in \bR$. Note that the random variable $(\cdot, f)_\sharp$ is equal to a Gaussian random variable in distribution, with mean 0 and variance $\lambda\parallel f \parallel^2 \equiv \lambda \langle f, f \rangle$.

For $f \in B^\ast_\bC \equiv B^\ast \otimes_\bR \bC \subset H \otimes_\bR \bC \equiv H_\bC$, $(\cdot, f)_\sharp$ can be written as an infinite sum of random variables $\sum_{n=1}^\infty c_n(\cdot, f_n)_\sharp$, using complex coefficients $c_n$, each such real random variable $(\cdot, f_n)_\sharp$ on $(B, \mu)$ is equal to a Gaussian random variable in distribution.

Write
\begin{align*}
\left( \cdot, f \right)_\sharp =& \sqrt{-1}\left( \cdot, g_1 \right)_\sharp + \left( \cdot, g_2 \right)_\sharp = \left[\frac{\langle g_1, g_2 \rangle}{\langle g_1, g_1 \rangle} + \sqrt{-1}\right]\left( \cdot, g_1 \right)_\sharp + \left( \cdot, g_2 - \frac{\langle g_1, g_2 \rangle}{\langle g_1, g_1 \rangle}g_1 \right)_\sharp.
\end{align*}
Thus, $|(\cdot, f)_\sharp|^2$ is equal in distribution to $|Z_1 + \sqrt{-1}Z_2|^2$, whereby each $Z_i$ is a (real) Gaussian random variable with mean 0. Extend $\langle \cdot , \cdot \rangle$ to be a sesquilinear inner product on the complexified Hilbert space $H_\bC$. 

Let $H_i \equiv (H_i, \langle \cdot, \cdot \rangle_i)$ be Hilbert spaces, $i=1,2$. When we write $H_1 \otimes H_2$, we define it as the tensor product of Hilbert spaces, with the tensor inner product \beq \langle f_1 \otimes g_1, f_2 \otimes g_2 \rangle = \langle f_1, f_2 \rangle_1 \langle g_1 , g_2 \rangle_2. \nonumber \eeq Now suppose $H_2$ is finite dimensional with dimension $n$. Then, $H_1 \otimes H_2 \cong H_1^{\times^n}$, direct product of $n$ copies of $H_1$. Suppose $B_1^\ast \subset H_1 \subset B_1$ is an Abstract Wiener space containing $H_1$ as a dense subspace, with Wiener measure $\tilde{\mu}$. We will construct an Abstract Wiener space $B_1 \otimes H_2$, containing $H_1 \otimes H_2$, with the product Wiener measure $\tilde{\mu}^{\times^n} \equiv \tilde{\mu} \times \cdots \times \tilde{\mu}$, $n$ copies of $\tilde{\mu}$.

When we write $B^{\otimes^k}$, we mean the $k$-th tensor product of $B$. An element $x$ in $B^{\otimes^k}$ can be written in the form $\sum_\alpha x_1^\alpha \otimes x_2^\alpha \otimes \cdots \otimes x_k^\alpha$, $x_i^\alpha \in B$. Likewise, an element $z \in B^{\otimes^k, \ast}$ can be written in the form $\sum_\beta z_1^\beta \otimes z_2^\beta \otimes \cdots \otimes z_k^\beta$, where $z_i^\beta \in B^{\ast}$.

When we write $(x, z)_\sharp$, we mean \beq \sum_\alpha \sum_\beta (x_1^\alpha \otimes x_2^\alpha \otimes \cdots \otimes x_k^\alpha, z_1^\beta \otimes z_2^\beta \otimes \cdots \otimes z_k^\beta )_\sharp = \sum_\alpha \sum_\beta (x_1^\alpha, z_1^\beta)_\sharp (x_2^\alpha, z_2^\beta)_\sharp \cdots (x_k^\alpha, z_k^\beta)_\sharp. \nonumber \eeq

\subsection{Triple and quartic terms}\label{ss.tqt}

The reference for the following remarks can be found in \cite{YMLim01}.

\begin{rem}\label{r.rem.1}
\begin{enumerate}
  \item We define a linear operator $\fd_a: H^2(\bC^4) \rightarrow H^2(\bC^4)$, for $a=0, 1, 2, 3$, whereby \beq \fd_a \left[ z_a^p \prod_{b \neq a \atop b=0,\cdots 3}z_b^{q_b} \right] = \left[\frac{p}{2}z_a^{p-1} - \frac{1}{2}z_a^{p+1}\right] \cdot \prod_{b \neq a \atop b=0,\cdots 3}z_b^{q_b},\ \ p, q_b \in \mathbb{N} \cup \{0\}. \nonumber \eeq We denote the range of $\fd_a$ by $\fd_a H^2(\bC^4)$. With this operator, we can define an operator $\fd: H^2(\bC^4)\otimes \Lambda^1(\bR^3 ) \otimes \mathfrak{g} \rightarrow H^2(\bC^4)\otimes \Lambda^2(\bR^4) \otimes \mathfrak{g}$, analogous to the differential operator $d$, as
      \begin{align}
      \fd A =& \sum_{\alpha=1}^N\sum_{i=1}^3 \fd_0 A_{i,\alpha} \otimes dx^0 \wedge dx^i \otimes E^\alpha \nonumber \\
      &+ \sum_{\alpha=1}^N\sum_{1 \leq i < j \leq 3}[\fd_i A_{j,\alpha} - \fd_j A_{i,\alpha}] \otimes dx^i \wedge dx^j \otimes E^\alpha,\label{e.b.1}
      \end{align}
      for $A = \sum_{\alpha=1}^N \sum_{i=1}^3 A_{i,\alpha} \otimes dx^i \otimes E^\alpha$, and $A_{i,\alpha} \in H^2(\bC^4)$ for $i=1,2,3$.
  \item\label{i.sb.2} Recall we had a renormalization flow \beq \left\{\Psi_\kappa: \mathcal{S}_\kappa(\bR^4) \rightarrow H^2(\bC^4) \ |\ \kappa>0 \right\}. \nonumber \eeq Because $\Psi_\kappa[\partial_a f] = \kappa \fd_a \Psi_\kappa[f]$, hence we need to apply a renormalization rule, which assigns a $\kappa$ to both $\fd$ and $\fd_a$, when defining our path integral.
  \item In \cite{YMLim01}, we defined an $L^2$ inner product $\langle f, g \rangle$, for $f, g \in \mathcal{S}_\kappa(\bR^4)$. We can extend this inner product into a tensor inner product, for $\mathfrak{g}$-valued 2-forms, over in $\bR^4$. For $A = \sum_{\alpha=1}^N\sum_{i=1}^3 a_{i,\alpha} \otimes dx^i \otimes E^\alpha$, the tensor inner product is \beq \langle dA, dA \rangle = \sum_{\gamma=1}^N\sum_{1\leq i<j \leq 3}\int_{\bR^4} a_{i;j, \gamma}^2\ d\omega + \sum_{\gamma=1}^N\sum_{1\leq j\leq 3}\int_{\bR^4} a_{0:j,\gamma}^2\ d\omega, \nonumber \eeq on the inner product space $\mathcal{S}_\kappa(\bR^4)\otimes \Lambda^2(\bR^4) \otimes \mathfrak{g}$. See Equation (\ref{e.d.1}). We extend the definition of Segal Bargmann Transform, so \beq \Psi_\kappa[A] = \sum_{\alpha=1}^N\sum_{i=1}^3 \Psi_\kappa[a_{i,\alpha}] \otimes dx^i \otimes E^\alpha. \nonumber \eeq
\item For \beq F = \sum_{\alpha=1}^N \sum_{i=1}^3f_{i,\alpha} \otimes dx^i \otimes E^\alpha,\quad G = \sum_{\alpha=1}^N \sum_{i=1}^3g_{i,\alpha} \otimes dx^i \otimes E^\alpha \nonumber \eeq
    in $H^2(\bC^4) \otimes \Lambda^1(\bR^3 ) \otimes \mathfrak{g}$, we define a tensor inner product
\begin{align}
\langle \fd F, \fd G \rangle =& \sum_{\alpha=1}^N \sum_{i=1}^3 \int_{\bC^4} d\lambda_4\ [\fd_0 f_{i,\alpha}]\overline{[\fd_0 g_{i,\alpha}]} \nonumber \\
&+ \sum_{\alpha=1}^N\sum_{1\leq i < j \leq 3} \int_{\bC^4} d\lambda_4\ [\fd_i f_{j,\alpha} - \fd_j f_{i,\alpha} ]\overline{[\fd_i g_{j,\alpha} - \fd_j g_{i,\alpha} ]}. \label{e.inn.1}
\end{align}
Note that $|g|^2 \equiv g \cdot \bar{g}$, $\bar{g}$ is the complex conjugate of a holomorphic function $g$.
\item\label{i.sb.1} The Segal Bargmann Transform $\Psi_\kappa$ is an isometry (up to a constant $\kappa^2$) between their respective spaces, i.e. \beq \langle dA, dA \rangle = \kappa^2 \langle \fd \Psi_{\kappa}[A], \fd \Psi_\kappa[A] \rangle, \nonumber \eeq $A \in \mathcal{S}_\kappa(\bR^4) \otimes \Lambda^1(\bR^3 ) \otimes \mathfrak{g}$.
\item\label{i.sb.3} We defined a closed subspace $B_0(\bC^4) \subset B(\bC^4)$, using the norm $\parallel x \parallel := |x| + |\fd_0 x|$, which can be shown to be a measurable norm. Hence, we can equip $B_0(\bC^4)$ with a Wiener measure. The map $\fd_0$ can be shown to be an injective map on $B_0(\bC^4)$ and we can define a measurable norm on its range $\fd_0 H^2(\bC^4) \subset \fd_0 B_0(\bC^4) \subset B(\bC^4)$ by $\parallel \fd_0 x \parallel := |x| + |\fd_0 x|$. Hence $\fd_0$ is an isometry map and thus $(\fd_0 B_0(\bC^4), \parallel \cdot \parallel)$ is also a Banach space, equipped with a Wiener measure.
\end{enumerate}
\end{rem}

Let the span of $\{dx^0 \wedge dx^i, 1\leq i \leq 3\}$ and the span of $\{dx^i \wedge dx^j, 1\leq i < j\leq 3\}$ be denoted by $\ast \Lambda^2(\bR^3)$ and $\Lambda^2(\bR^3)$ respectively. Define the inner product space \beq \mathbb{H}:= \left\{[\fd_0 H^2(\bC^4)] \otimes \left[\ast \Lambda^2(\bR^3) \right] \right\} \oplus \left\{ H^2(\bC^4) \otimes \Lambda^2(\bR^3) \right\}. \nonumber \eeq

Observe that for $A \in H^2(\bC^4) \otimes \Lambda^1(\bR^3)$, we have that
\begin{align*}
\sum_{i=1}^3 \fd_0 A_i \otimes dx^0 \wedge dx^1 \in& [\fd_0 H^2(\bC^4)] \otimes \left[\ast \Lambda^2(\bR^3) \right], \\
\sum_{1\leq i < j \leq 3}[\fd_i A_j - \fd_j A_i]\otimes dx^i \wedge dx^j \in& H^2(\bC^4) \otimes \Lambda^2(\bR^3) ,
\end{align*}
and the said inner product given by Equation (\ref{e.inn.1}).

We can complete it into a Banach space
\beq  \dB:= \left\{[\fd_0 B_0(\bC^4)] \otimes [\ast\Lambda^2(\bR^3)] \right\}\oplus \left\{ B_0(\bC^4) \otimes \Lambda^2(\bR^3)\right\} , \nonumber \eeq using a measurable norm.

\begin{rem}\label{r.w.1}
It was proved in \cite{YMLim01}, that this Banach space contains the dense inner product space $H^2(\bC^4) \otimes \Lambda^2(\bR^4)$, which allows us to define a product Wiener measure on $\dB$.
\end{rem}

Just like how Balaban in his series of papers, defined a finite dimensional integral on a finite gauge lattice of spacing $\epsilon$, we will define a path integral over in $\mathcal{S}_\kappa(\bR^4) \otimes \Lambda^1(\bR^3) \otimes \mathfrak{g}$ from Expression \ref{e.ym.1} as \beq \frac{1}{Z}\Tr\int_{\{d A:\ A \in \mathcal{S}_\kappa(\bR^4) \otimes \Lambda^1(\bR^3) \otimes \mathfrak{g}\}} \exp \left[  c\int_{R[a,T]} d[A^\rho] \right] e^{-\frac{1}{2}S_{{\rm YM}}(A)}\ D[dA], \label{e.ym.3} \eeq $Z$ is some normalization constant. Note that $\epsilon$ is synonymous with $1/\kappa$ and $\Tr$ means taking the matrix trace.

After Balaban defined a finite dimensional integral, he applied a renormalization flow and mapped the finite gauge lattice of spacing $\epsilon$ to a finite gauge lattice of unit spacing. We will do the same thing here, but using the Segal Bargmann Transform $\Psi_\kappa$ instead.

\begin{rem}
We can also view Segal Bargmann Transform as analogous to Fourier Transform, which is how physicists will attempt to define a path integral over in momentum space.
\end{rem}

Because of the isometry given by Item \ref{i.sb.1} in Remark \ref{r.rem.1}, we will make sense of \beq \frac{1}{Z}\exp\left[ -\frac{1}{2}\int_{\bR^4}d\omega\ \left(\sum_{1\leq i<j \leq 3} a_{i;j}^2 +  \sum_{j=1}^3 a_{0:j}^2 \right) \right]D[dA]\nonumber \eeq as
\beq \frac{1}{Z}\exp\left(-\frac{1}{2} \int_{\bC^4}d\lambda_4 |\kappa\fd A|^2 \right)D[\fd A] \label{e.a.2} \eeq
over in $\mathbb{H}$, instead of over in $\mathcal{S}_\kappa(\bR^4) \otimes \Lambda^2(\bR^4 )$. For each $\kappa$, we will interpret Expression \ref{e.a.2} as a Wiener measure $\tilde{\mu}_{\kappa^2}$ with variance $1/\kappa^2$, over the Banach space $\dB$, containing $\mathbb{H}$. Note that $(\dB, \tilde{\mu}_{\kappa^2})$ is actually a probability space, the $\sigma$ algebra is the Borel field of $\dB$.

Refer to the isometry given by Item \ref{i.sb.1} in Remark \ref{r.rem.1}. We will extend this isometry to a product of triple or quartic terms. Therefore, we will make sense of $\frac{1}{Z}e^{-\frac{1}{2}S_{{\rm YM}}(A)}\ D[dA]$ in Expression \ref{e.ym.3}, by replacing it with
\begin{align*}
\frac{1}{Z}\exp&\left( -\frac{1}{2}  \int_{\bC^4}d\lambda_4 |\kappa\fd A + A \wedge A|^2 \right)D[\fd A] = D_1 D_2,  \\
\end{align*}
whereby
\begin{align}
\label{e.x.1}
\begin{split}
D_1 &:= \exp \left( -\frac{1}{2}\int_{\bC^4}d\lambda_4 \langle \kappa\fd A, A \wedge A \rangle + \langle A \wedge A, \kappa\fd A \rangle + |A \wedge A|^2 \right) , \\
D_2 &:= \frac{1}{Z}\exp\left(-\frac{1}{2} \int_{\bC^4}d\lambda_4 |\kappa\fd A|^2 \right)D[\fd A],
\end{split}
\end{align}
all on $\mathbb{H} \otimes \mathfrak{g}$, $Z$ is some normalization constant.

For $A = \sum_{\alpha=1}^N \sum_{i=1}^3 A_{i,\alpha} \otimes dx^i \otimes E^\alpha \in H^2(\bC^4) \otimes \Lambda^1(\bR^3) \otimes \mathfrak{g}$, we have $\fd A$ as given in Equation (\ref{e.b.1}), \beq A \wedge A = \sum_{\alpha, \beta=1}^N\sum_{1 \leq i < j \leq 3} A_{i,\alpha}A_{j,\beta} \otimes dx^i \wedge dx^j \otimes [E^\alpha, E^\beta], \nonumber \eeq which can be written as \beq A \wedge A = \sum_{\gamma=1}^N \sum_{\alpha, \beta=1}^N\sum_{1 \leq i < j \leq 3} c_\gamma^{\alpha\beta}A_{i,\alpha}A_{j,\beta} \otimes dx^i \wedge dx^j \otimes E^\gamma. \nonumber \eeq

Therefore, \beq |A \wedge A|^2 = \sum_{\gamma=1}^N\sum_{1 \leq i < j \leq 3}\sum_{\alpha, \beta \atop \hat{\alpha}, \hat{\beta}} c_\gamma^{\alpha\beta}c_\gamma^{\hat{\alpha}\hat{\beta}} A_{i,\alpha}A_{j,\beta}\overline{A_{i,\hat{\alpha}}}
\overline{A_{j,\hat{\beta}}}, \nonumber \eeq
and
\begin{align*}
\langle \fd A , A \wedge A \rangle =& \sum_{\gamma=1}^N \sum_{1 \leq i < j \leq 3} \sum_{\alpha, \beta=1}^Nc_\gamma^{\alpha\beta}[\fd_i A_{j,\gamma} - \fd_j A_{i,\gamma}]\overline{A_{i,\alpha}A_{j,\beta}}, \\
\langle A \wedge A , \fd A  \rangle =& \sum_{\gamma=1}^N \sum_{1 \leq i < j \leq 3} \sum_{\alpha, \beta=1}^Nc_\gamma^{\alpha\beta}A_{i,\alpha}A_{j,\beta}\overline{[\fd_i A_{j,\gamma} - \fd_j A_{i,\gamma}]}.
\end{align*}

But we would like to interpret them as tensor products of complex random variables on the Banach space $\dB \otimes \mathfrak{g}$, which contains $\mathbb{H} \otimes \mathfrak{g}$. To do this, we first need to review one linear functional, defined in \cite{YMLim01}.

\begin{defn}
For any $f \in \mathcal{H}^2(\bC^4)$, $\chi_w(z) \equiv e^{\bar{w}\cdot z}$ is the unique vector such that $\langle f, \chi_w \rangle = f(w)$ for any $w \in \bC^4$. Note that $\bar{w} \cdot z = \sum_{a=0}^3 \overline{w_a} z_a$, $\overline{w_a}$ is the complex conjugate of $w_a \in \bC$. Because $\chi_w \in B(\bC^4)_\bC^\ast \subset \mathcal{H}(\bC^4)_\bC$, we see that it is also in $B_0(\bC^4)_\bC^\ast$ and $[\fd_0 B_0(\bC^4)]_\bC^\ast$. Thus, $(\cdot, \chi_w)_\sharp$ is a complex random variable.

Define \beq \tilde{\xi}_{ab}(w):= (-1)^{ab}\chi_w \otimes dx^a \wedge dx^b,\ 0\leq a < b \leq 3, \nonumber \eeq and $\tilde{\xi}_{ab,\gamma}(w) := (-1)^{ab} \tilde{\xi}_{ab}(w) \otimes E^\gamma$.
\end{defn}

For $A = \sum_{i=1}^3 A_i \otimes dx^i$, this linear functional $(\cdot, \tilde{\xi}_{ab}(w))_\sharp$ sends  $\fd A$ to \beq (-1)^{ab}[\fd_a A_b - \fd_b A_a](w) = \langle \fd A(w), (-1)^{ab} dx^a \wedge dx^b \rangle_2 = (\fd A, \tilde{\xi}_{ab}(w))_\sharp, \nonumber \eeq $A_0 \equiv 0$. Refer to \cite{YMLim01}.

And how are we going to make sense of $A_i(w)$? In \cite{YMLim01}, we showed that the map $\fd_0: B_0(\bC^4) \rightarrow B(\bC^4)$ is a bounded and injective map.

\begin{prop}\label{p.z.1}
For each $w \in \bC^4$, we can define a bounded linear functional $f_w: \fd_0 B_0(\bC^4) \rightarrow \bC$, \beq f_w(\fd_0 x) := (x, \chi_w)_\sharp, \nonumber \eeq and thus there is a $\zeta(w) \in [\fd_0 B_0(\bC^4)]_\bC^\ast \subset \mathcal{H}(\bC^4)_\bC$ such that \beq f_w(\fd_0x) = (\fd_0 x, \zeta(w))_\sharp = (x, \chi_w)_\sharp = x(w). \nonumber \eeq
\end{prop}

\begin{proof}
Note that $(\cdot, \chi_w)_\sharp: (B(\bC^4), |\cdot|) \rightarrow \bC$ is a bounded linear functional. For any $y \in \fd_0  B_0(\bC^4) \subset B(\bC^4)$, we can find an unique $x \in B_0(\bC^4)$ such that $y = \fd_0 x$. Hence,
\begin{align*}
|f_w(y)| =& |f_w(\fd_0 x)| = |x(w)| = |(x, \chi_w)_\sharp| \\
\leq& C|x|  \leq C\parallel x \parallel = C\parallel y \parallel.
\end{align*}
This shows that $f_w$ is a bounded linear functional on $\fd_0 B_0(\bC^4)$.
\end{proof}

\begin{defn}
Write $\zeta_i(w) := \zeta(w) \otimes dx^0 \wedge dx^i$, $i=1,2,3$. From the previous proposition, $\zeta_i(w)$ lies inside $\dB_\bC^\ast$. Note that $(\cdot, \zeta_i(w))_\sharp$ is a complex random variable on the Banach space and for $A = \sum_{i=1}^3 A_i \otimes dx^i \in H^2(\bC^4) \otimes \Lambda^1(\bR^3)$, \beq \left( \cdot, \zeta_i(w) \right)_\sharp:\ \fd A \longmapsto A_i(w) \equiv \left( \fd_0 A_i, \zeta(w) \right)_\sharp. \nonumber \eeq
\end{defn}

For $w \in \bC^4$, we will write $A_{i}(w) = (\fd A, \zeta_i(w))_\sharp$. Observe that we can now write the product \beq [A_iA_j](w) = A_i(w) A_j(w) = \left(\fd A \otimes \fd A, \zeta_i(w) \otimes \zeta_j(w) \right)_\sharp. \nonumber \eeq And, the complex conjugate \beq   \overline{(\fd A, \zeta_i(w))_\sharp} = \overline{A_i(w)} = A_i(\bar{w}) = (\fd A, \zeta_i(\bar{w}))_\sharp. \label{e.c.3} \eeq

\begin{rem}\label{e.r.1}
Note that $A_i \in B_0(\bC^4)$, which consists of holomorphic functions over $\bC^4$, with real coefficients. Hence, the product $A_iA_j$ is actually well defined and Equation (\ref{e.c.3}) holds.
\end{rem}

\begin{notation}
To ease our notations, we will write for $j=1,2,3$, $\zeta_{j,w} \equiv \zeta_j(w)$ and $\tilde{\pi}_{j, \alpha, w}= \tilde{\pi}_{j,\alpha}(w) = \zeta_{j}(w) \otimes E^{\alpha}$, $\{E^\alpha\}_{\alpha=1}^N$ is an orthonormal basis in $\mathfrak{g}$. So \beq \left( \fd A , \tilde{\pi}_{i,\alpha,w} \right)_\sharp = \left( \fd_0A_{i, \alpha},  \zeta(w) \right)_\sharp . \nonumber \eeq Note that $A = \sum_{\alpha=1}^N \sum_{i=1}^3 A_{i,\alpha} \otimes dx^i \otimes E^\alpha \in B_0(\bC^4) \otimes \Lambda^1(\bR^3) \otimes \mathfrak{g}$.
\end{notation}

We can now explain how to define the quartic term, given by $A_{i,\alpha}A_{j,\beta}\overline{A_{i,\hat{\alpha}}}\overline{A_{j,\hat{\beta}}}$, whereby $A_{i,\alpha}$ is a holomorphic function over $\bC^4$ and $\overline{A_{i,\alpha}}$ denotes its complex conjugate. The trick is to write for each $A_{i,\alpha} \in H^2(\bC^4)$,
\begin{align*}
[A_{i,\alpha}&A_{j,\beta}\overline{A_{i,\hat{\alpha}}}\overline{A_{j,\hat{\beta}}}](w) \\
&=
\langle \fd A, \zeta_i(w)\otimes E^\alpha \rangle\ \langle \fd A, \zeta_j(w)\otimes E^\beta \rangle\ \overline{\langle \fd A, \zeta_i(w) \otimes E^{\hat{\alpha}} \rangle}\ \overline{ \langle \fd A, \zeta_j(w) \otimes E^{\hat{\beta}} \rangle} \\
&=: \left\langle \fd A \otimes \fd A \otimes \fd A \otimes \fd A, \tilde{\pi}_{i,\alpha,w} \otimes \tilde{\pi}_{j,\beta,w} \otimes \tilde{\pi}_{i,\hat{\alpha},\bar{w}} \otimes \tilde{\pi}_{j,\hat{\beta},\bar{w}} \right\rangle ,
\end{align*}
for $\fd A$ defined in Equation (\ref{e.b.1}).

Therefore,
\begin{align*}
\int_{w \in \bC^4 }&[A_{i,\alpha}A_{j,\beta}\overline{A_{i,\hat{\alpha}}A_{j,\hat{\beta}}}](w)
d\lambda_4(w) \\
&=\int_{w \in \bC^4}\ \left\langle \fd A^{\otimes^4}, \tilde{\pi}_{i,\alpha,w} \otimes \tilde{\pi}_{j,\beta,w} \otimes \tilde{\pi}_{i,\hat{\alpha},\bar{w}} \otimes \tilde{\pi}_{j,\hat{\beta},\bar{w}} \right\rangle\ d\lambda_4(w) \\
&= \left\langle \fd A^{\otimes^4}, \int_{w \in \bC^4}d\lambda_4(w)\ \tilde{\pi}_{i,\alpha,w} \otimes \tilde{\pi}_{j,\beta,w} \otimes \tilde{\pi}_{i,\hat{\alpha},\bar{w}} \otimes \tilde{\pi}_{j,\hat{\beta},\bar{w}} \right\rangle .
\end{align*}
This means that we interpret the quartic term in the Yang-Mills integral, as a tensor product of functionals. A similar approach will be used to define the cubic term.

For $A = \sum_{\alpha=1}^N\sum_{i=1}^3 A_{i,\alpha} \otimes dx^i \otimes E^\alpha \in B_0(\bC^4) \otimes \Lambda^1(\bR^3) \otimes \mathfrak{g}$, we will interpret $\int_{\bC^4}d\lambda_4|A\wedge A|^2$ as
\begin{align}
\sum_{\gamma=1}^N&\sum_{1\leq i<j \leq 3}  \sum_{\alpha , \beta \atop \hat{\alpha} , \hat{\beta}} c_\gamma^{\alpha \beta}c_\gamma^{\hat{\alpha} \hat{\beta}} \int_{w \in \bC^4}d\lambda_4(w)\ \bigg[\left( \fd A^{\otimes^4}, \tilde{\pi}_{i, \alpha} \otimes \tilde{\pi}_{j,\beta} \otimes \tilde{\pi}_{i,\hat{\alpha}} \otimes \tilde{\pi}_{j,\hat{\beta}}(w) \right)_{\sharp,34} \bigg] \nonumber \\
\equiv&\sum_{\gamma=1}^N\sum_{1\leq i<j \leq 3}  \sum_{\alpha , \beta \atop \hat{\alpha} , \hat{\beta}} c_\gamma^{\alpha \beta}c_\gamma^{\hat{\alpha} \hat{\beta}} \bigg[\left( \fd A^{\otimes^4}, \int_{w \in \bC^4}d\lambda_4(w)\tilde{\pi}_{i, \alpha} \otimes \tilde{\pi}_{j,\beta} \otimes \tilde{\pi}_{i,\hat{\alpha}} \otimes \tilde{\pi}_{j,\hat{\beta}}(w)  \right)_{\sharp,34} \bigg], \label{e.x.6}
\end{align}
where the complex-valued random variable on $\left(\dB \otimes \mathfrak{g}, \tilde{\mu}_{\kappa^2}^{\times^{N}}\right)$ is defined as
\begin{align}
\left(B^{\otimes^4}, \tilde{\pi}_{i_1, \alpha_1} \otimes \tilde{\pi}_{i_2,\alpha_2} \otimes \tilde{\pi}_{i_3,\alpha_3} \otimes \tilde{\pi}_{i_4,\alpha_4}(w) \right)_{\sharp, 34} :=& \prod_{j=1}^2\left(B, \tilde{\pi}_{i_j,\alpha_j,w}\right)_\sharp \times \prod_{j=3}^4\left(B, \tilde{\pi}_{i_j,\alpha_j,\bar{w}}\right)_\sharp.  \nonumber
\end{align}

Similarly, we interpret $\int_{\bC^4}d\lambda_4 \left\langle \kappa\fd A, A \wedge A \right\rangle$ and $\int_{\bC^4}d\lambda_4 \left\langle A \wedge A, \kappa\fd A \right\rangle$ as
\begin{align}
\label{e.x.2}
\begin{split}
\sum_{\gamma=1}^N&\sum_{1\leq i < j \leq 3}\sum_{\alpha , \beta} c_\gamma^{\alpha\beta}\left( \fd A^{\otimes^3}, \int_{w \in \bC^4}d\lambda_4(w) \kappa\tilde{\xi}_{ij, \gamma}  \otimes \tilde{\pi}_{i,\alpha}\otimes \tilde{\pi}_{j,\beta}(w) \right)_{\sharp,23}, \\
\sum_{\gamma=1}^N&\sum_{1\leq i < j \leq 3}\sum_{\alpha , \beta}c_\gamma^{\alpha\beta}\left( \fd A^{\otimes^3} , \int_{w \in \bC^4}d\lambda_4(w)\ \tilde{\pi}_{i,\alpha} \otimes \tilde{\pi}_{j,\beta} \otimes \kappa\tilde{\xi}_{ij,\gamma}(w)  \right)_{\sharp,3},
\end{split}
\end{align}
respectively, where the respective complex-valued random variables on $\left(\dB \otimes \mathfrak{g}, \tilde{\mu}_{\kappa^2}^{\times^{N}} \right)$ are defined as
\begin{align*}
\left(B^{\otimes^3}, \kappa\tilde{\xi}_{ij, \alpha_1}  \otimes \tilde{\pi}_{i_2,\alpha_2}\otimes \tilde{\pi}_{i_3,\alpha_3}(w) \right)_{\sharp, 23} :=& \kappa\left(B, \tilde{\xi}_{ij}(w)\otimes E^{\alpha_1}\right)_\sharp \times \prod_{j=2}^3 \left(B, \tilde{\pi}_{i_j,\alpha_j,\bar{w}}\right)_\sharp, \\
\left(B^{\otimes^3}, \tilde{\pi}_{i_1,\alpha_1}\otimes \tilde{\pi}_{i_2,\alpha_2} \otimes \kappa\tilde{\xi}_{ij,\alpha_3}(w) \right)_{\sharp, 3} :=& \prod_{j=1}^2\left(B, \tilde{\pi}_{i_j,\alpha_j,w}\right)_\sharp \times  \kappa \left(B, \tilde{\xi}_{ij}(\bar{w})\otimes E^{\alpha_3}\right)_{\sharp}.
\end{align*}
Note that from Equation (\ref{e.b.1}),
\begin{align*}
\sum_{\alpha=1}^N \fd_0 A_{i,\alpha} \otimes dx^0 \wedge dx^i \otimes E^\alpha &\in \fd_0 B_0(\bC^4) \otimes (dx^0 \wedge dx^i)\otimes \mathfrak{g} \subset \dB\otimes \mathfrak{g}, \\
[\fd_i A_{j,\alpha} - \fd_j A_{i,\alpha}]\otimes dx^i \wedge dx^j \otimes E^\alpha &\in B_0(\bC^4) \otimes (dx^i \wedge dx^j)\otimes \mathfrak{g} \subset \dB\otimes \mathfrak{g}.
\end{align*}

\begin{rem}
For $B = \sum_{\alpha=1}^N\sum_{0\leq a < b \leq 3}B_{ab,\alpha} \otimes dx^a \wedge dx^b \otimes E^\alpha \in \dB \otimes \mathfrak{g}$, $B_{ab,\alpha}$ is holomorphic. In fact, $\left(B, \tilde{\xi}_{ij}(w)\otimes E^\alpha\right)_\sharp := (-1)^{ij}B_{ij,\alpha}(w)$ and its complex conjugate $\overline{\left(B, \tilde{\xi}_{ij}(w)\otimes E^\alpha\right)_\sharp} = (-1)^{ij}B_{ij,\alpha}(\bar{w})$. And if $B_{0i,\alpha} = \fd_0 \tilde{B}_{i,\alpha}$ for $i=1,2,3$, then
\begin{align*}
\left(B, \tilde{\pi}_{i_j,\alpha_j,w}\right)_\sharp := \left(\fd_0 \tilde{B}_{i_j,\alpha_j}, \zeta(w) \right) = \tilde{B}_{i_j,\alpha_j}(w).
\end{align*}
\end{rem}

But just like in the abelian case considered in \cite{YMLim01}, we do not consider $\tilde{\xi}_{ab}(w)$ and $\tilde{\pi}_{i,\alpha}(w)$. Instead, we will introduce a renormalization factor $\psi_w$, defined in Notation \ref{n.n.5}. That is we replace $\tilde{\xi}_{ab}(w)$ and $\tilde{\pi}_{i,\alpha}(w)$ with \beq \xi_{ab}(w) = \psi_w\tilde{\xi}_{ab}(w)\ \ {\rm and}\ \ \pi_{i,\alpha}(w) \equiv \pi_{i,\alpha,w} = \psi_w \tilde{\pi}_{i,\alpha,w} \nonumber \eeq respectively.

\begin{rem}\label{r.rem.2}
\begin{enumerate}
  \item We explained in \cite{YMLim01}, the necessity of this renormalization factor $\psi_w$.
  \item A direct computation will show that $\parallel \xi_{ab}(w) \parallel^2 = \langle \xi_{ab}(w), \xi_{ab}(w) \rangle = 1/2\pi$. This is because both $B_0(\bC^4)$ and $\fd_0 B_0(\bC^4)$ contain a complete set of basis in $\mathcal{H}^2(\bC^4)$. See Remark \ref{r.w.1}.
  \item In \cite{YMLim01}, we showed that $\langle f, f \rangle \leq 40 \langle \fd_0 f, \fd_0 f \rangle$, $f \in H^2(\bC^4)$. By definition of $\zeta(w)$ in Proposition \ref{p.z.1}, we see that \beq \parallel  \pi_{i,\alpha}(w) \parallel^2 =  \langle \pi_{i,\alpha}(w), \pi_{i,\alpha}(w) \rangle \leq \frac{40}{2\pi}. \nonumber \eeq
\end{enumerate}
\end{rem}

\begin{defn}\label{d.ym.3}
Apply the renormalization rule and add in the renormalization factor. For $B \in \dB \otimes \mathfrak{g}$, we will define
\begin{align*}
Y_1^\kappa & \left( B\right) \\
&:= \sum_{\gamma=1}^N\sum_{1\leq i < j \leq 3}\sum_{\alpha , \beta}c_\gamma^{\alpha\beta}\left( B^{\otimes^3}, \int_{w \in \bC^4}d\lambda_4(w)\ \kappa\xi_{ij, \gamma}  \otimes \pi_{i,\alpha}\otimes \pi_{j,\beta}(w) \right)_{\sharp,23}, \\
Y_2^\kappa &  \left( B\right) \\
&:= \sum_{\gamma=1}^N\sum_{1\leq i < j \leq 3}\sum_{\alpha , \beta}c_\gamma^{\alpha\beta}\left( B^{\otimes^3}, \int_{w \in \bC^4}d\lambda_4(w)\ \pi_{i,\alpha} \otimes \pi_{j,\beta} \otimes \kappa\xi_{ij,\gamma}(w) \right)_{\sharp,3},
\end{align*}
and
\begin{align*}
Y_3^\kappa & \left( B\right) \\
&:= \sum_{\gamma=1}^N\sum_{1\leq i<j \leq 3}  \sum_{\alpha , \beta \atop \hat{\alpha} , \hat{\beta}} c_\gamma^{\alpha \beta}c_\gamma^{\hat{\alpha} \hat{\beta}} \left( B^{\otimes^4}, \int_{w \in \bC^4}d\lambda_4(w)\ \pi_{i, \alpha} \otimes \pi_{j,\beta} \otimes \pi_{i,\hat{\alpha}} \otimes \pi_{j,\hat{\beta}}(w) \right)_{\sharp,34} ,
\end{align*}
all being random variables on $\left(\dB \otimes \mathfrak{g}, \tilde{\mu}_{\kappa^2}^{\times^{N}}\right)$, using Expressions \ref{e.x.6} and \ref{e.x.2}.

From Expression \ref{e.x.1}, we replace $\fd A$ with $B$, so we interpret the expression over in $\dB \otimes \mathfrak{g}$. Using the construction of the Abstract Wiener space in \cite{YMLim01}, we will interpret \beq D_2 = \frac{1}{Z}e^{-\frac{1}{2}\int_{\bC^4} d\lambda_4 |\kappa \fd A|^2}D[\fd A] \nonumber \eeq as product Wiener measure $\tilde{\mu}_{\kappa^2}^{\times^{N}}$ over $\dB \otimes \mathfrak{g}$, for some normalization constant $Z$.

And we interpret \beq D_1 = \exp \left( -\frac{1}{2}\int_{\bC^4}d\lambda_4 \langle \kappa\fd A, A \wedge A \rangle + \langle A \wedge A, \kappa\fd A \rangle + |A \wedge A|^2 \right) \nonumber \eeq as
\beq \mathcal{Y}^\kappa \left( B\right) :=  \exp\left\{ -\frac{1}{2}\left[ Y_1^\kappa  \left( B\right) + Y_2^\kappa\left( B\right) + Y_3^\kappa\left( B\right) \right] \right\}, \label{e.j.2} \eeq
\end{defn}
over in $\dB \otimes \mathfrak{g}$.

In Lemma \ref{l.y.1}, we will show that it is integrable with respect to product Wiener measure $\tilde{\mu}_{\kappa^2}^{\times^N}$.

\section{Yang-Mills Path Integral}\label{s.ymp}

In the abelian case considered in \cite{YMLim01}, we defined a Wiener (probability) measure $\tilde{\mu}_{\kappa^2}$, on a Banach space $\dB$. We will write $\bE$ to denote expectation on this probability space $(\dB, \tilde{\mu}_{\kappa^2})$.

In the non-abelian case, we will now construct a measure on $\dB \otimes \mathfrak{g}$, which is absolutely continuous with respect to the product Wiener measure $\tilde{\mu}_{\kappa^2}^{\times^{N}} \equiv \tilde{\mu}_{\kappa^2} \times \cdots \times \tilde{\mu}_{\kappa^2}$. That is, we will define a sequence of Yang-Mills measures on  \beq \dB \otimes \mathfrak{g},\quad {\rm as}\quad \left\{\mathcal{Y}^\kappa d\tilde{\mu}_{\kappa^2}^{\times^{N}}:\ \kappa>0\right\}. \nonumber \eeq This sequence of Yang-Mills measures is indexed by $\kappa$, which is due to the renormalization flow $\{\Psi_\kappa:\ \kappa > 0\}$ defined in Item \ref{i.sb.2} in Remark \ref{r.rem.1}. We will continue to write $\bE$ to denote expectation on this probability space $(\dB \otimes \mathfrak{g}, \tilde{\mu}_{\kappa^2}^{\times^{N}})$.

\begin{lem}\label{l.y.1}
Refer to Definition \ref{d.ym.3} for $\mathcal{Y}^\kappa$. Consider the probability space $\dB \otimes \mathfrak{g}$ equipped with a product Wiener measure $\tilde{\mu}_{\kappa^2}^{\times^{N}}$. For any $\kappa > 0$, \beq \mathbb{E}\left[\mathcal{Y}^\kappa \right] := \int_{\dB \otimes \mathfrak{g}} \mathcal{Y}^\kappa d\tilde{\mu}_{\kappa^2}^{\times^{N}}\nonumber \eeq is finite.
\end{lem}

\begin{proof}
Observe that $Y_1^\kappa  \left( B\right) + Y_2^\kappa\left( B\right) + Y_3^\kappa\left( B\right)$
is equal to
\begin{align*}
\sum_{\gamma=1}^N \int_{w\in\bC^4}d\lambda_4(w)\ \Bigg[ \sum_{1\leq i<j \leq 3}\Big|\left(B, \kappa\xi_{ij,\gamma}(w) \right)_\sharp + \sum_{\alpha , \beta}c_\gamma^{\alpha\beta}&\left(B^{\otimes^2}, \pi_{i,\alpha,w} \otimes \pi_{j, \beta, w} \right)_\sharp \Big|^2 \\
-& \sum_{1\leq i<j \leq 3}\left|\left( B, \kappa\xi_{ij, \gamma}(w)\right)_\sharp \right|^2 \Bigg].
\end{align*}
Thus,
\begin{align}
\mathcal{Y}^\kappa  \left( B\right)
=&\exp\Bigg[ -\frac{1}{2}\left[Y_1^\kappa  \left( B\right) + Y_2^\kappa\left( B\right) + Y_3^\kappa\left( B\right) \right]\Bigg] \nonumber \\
\leq& \exp\left[ \frac{1}{2}\int_{w \in \bC^4}d\lambda_4(w)\sum_{\alpha=1}^N\sum_{1\leq i<j \leq 3}\left|\left( B, \kappa\xi_{ij,\alpha}(w)\right)_\sharp\right|^2 \right], \nonumber
\end{align}
whereby $\{\left(\cdot, \kappa\xi_{ij,\alpha}(w)\right)_\sharp:\ 1\leq \alpha\leq N, 1\leq i < j \leq 3\}$ is a set of i.i.d. complex random variables with mean 0. Hence, it suffices to show \beq \exp\left[ \frac{1}{2}\int_{w \in \bC^4}d\lambda_4(w)\left|\left( \cdot, \kappa\xi_{ij}(w)\right)_\sharp\right|^2 \right]\nonumber \eeq is integrable on $\dB$.

Using Jensen's inequality, we have
\begin{align*}
\exp&\left[ \frac{1}{2}\int_{w \in \bC^4}d\lambda_4(w)\left|\left( \cdot, \kappa\xi_{ij}(w)\right)_\sharp \right|^2 \right] \\
&\leq
\int_{w \in \bC^4}d\lambda_4(w)\exp\left[\frac{1}{2}\left|\left( \cdot, \kappa\xi_{ij}(w)\right)_\sharp\right|^2 \right].
\end{align*}

Now, we can write the complex-valued random variable
\begin{align}
\left(\cdot, \kappa\xi_{ij,\alpha}(w)\right)_\sharp =& \sqrt{-1}\left(\cdot, f \right)_\sharp + \left(\cdot, g \right)_\sharp \nonumber \\
=& \left[\frac{\langle f, g \rangle}{\langle f, f\rangle} + \sqrt{-1} \right]\left(\cdot, f \right)_\sharp + \left(\cdot, g - \frac{\langle f, g \rangle}{\langle f, f \rangle}f \right)_\sharp, \label{e.gs.1}
\end{align}
whereby $f, g \in [\dB\otimes \mathfrak{g}]^\ast$. Thus $|\left(\cdot, \kappa\xi_{ij,\alpha}(w)\right)_\sharp|^2$ is equal to $|Z_1 + \sqrt{-1}Z_2|^2$ in distribution, each $Z_i$ is a real normal random variable with mean 0.

From Equation (\ref{e.gs.1}), write
\beq c = \frac{\langle f, g \rangle}{\langle f, f\rangle} + \sqrt{-1}, \quad W_1 = \left(\cdot, f \right)_\sharp, \quad W_2 = \left(\cdot, g - \frac{\langle f, g \rangle}{\langle f, f \rangle}f \right)_\sharp. \nonumber \eeq

Then $|(\cdot, \kappa \xi_{ij}(w))_\sharp| \leq | c| |W_1| + |W_2|$, so by Jenson's inequality, \beq |(\cdot, \kappa \xi_{ij}(w))_\sharp|^2 \leq 4\left[ \frac{1}{2}\left(| c| |W_1| + |W_2| \right) \right]^2 \leq 2[|c|^2|W_1|^2 + |W_2|^2]. \nonumber \eeq

From Remark \ref{r.rem.2}, $\bE[ |(\cdot, \kappa \xi_{ij}(w))_\sharp|^2] = 1/2\pi$. Because $W_1$ is independent of $W_2$, note that $\bE[|(\cdot, \kappa \xi_{ij}(w))_\sharp|^2] = |c|^2\bE[|W_1|^2] + \bE[|W_2|^2]$. Hence, for some constants $a^2 + b^2 = 1$, which depend on $w$, we have
\beq 2\pi|(\cdot, \kappa \xi_{ij}(w))_\sharp|^2 \leq 2[a^2\tilde{W}_1^2 + b^2\tilde{W}_2^2], \nonumber \eeq whereby for each $i=1,2$,
\begin{itemize}
  \item $\tilde{W}_i$ is a (real) random variable on $\left(\dB , \tilde{\mu}_{\kappa^2} \right)$,
  \item $\tilde{W}_i$ is equal in distribution to a standard normal $Z_i$ and
  \item $Z_1$ is independent of $Z_2$.
\end{itemize}

Using Fubini's Theorem,
\begin{align*}
\mathbb{E}&\left[\int_{w \in \bC^4}d\lambda_4(w)\exp\left[\frac{1}{2}\left|\left( \cdot, \kappa\xi_{ij}(w)\right)_\sharp\right|^2 \right] \right] \\
&= \int_{w \in \bC^4}d\lambda_4(w)\mathbb{E}\left[
\exp\left[\frac{1}{2}\left|\left( \cdot, \kappa\xi_{ij}(w)\right)_\sharp\right|^2 \right] \right]  \\
&\leq \int_{w \in \bC^4}d\lambda_4(w)\mathbb{E}\left[
\exp\left[\frac{1}{2\pi}[a^2\tilde{W}_1^2 + b^2\tilde{W}_2^2] \right] \right]
\leq \mathbb{E}\left[
\exp\left[\frac{1}{2\pi}[Z_1^2 + Z_2^2] \right] \right]\\
&= \frac{1}{2\pi}\int_{x\in \bR} e^{x^2 /2\pi}e^{-x^2/2}dx \cdot \int_{y\in \bR} e^{y^2 /2\pi}e^{-y^2/2}dy \\
&= \left[\frac{1}{\sqrt{1-(1/\pi)}} \right]^2 < \infty.
\end{align*}
\end{proof}

By modifying the above proof, we can actually prove the following result.

\begin{cor}\label{c.p.1}
For any $1\leq p < \pi$, we have \beq \mathbb{E} \left( [\mathcal{Y}^{\kappa}]^p \right) \leq \left(\frac{1}{\sqrt{1-p/\pi}} \right)^{2N}. \label{e.y.3} \eeq
\end{cor}

\begin{rem}\label{r.p.1}
If we replace the factor $1/\sqrt{2\pi}$ in $\psi_w$ with any $0<\tilde{c} < 1/\sqrt2$, then the corollary holds true for any $1 \leq p < 1/2\tilde{c}^2$.
\end{rem}

\begin{notation}(Casimir operator)\label{n.co.1}\\
Let $\mathfrak{g}$ be a semi-simple Lie Algebra. For an irreducible representation $\rho: \mathfrak{g} \rightarrow {\rm End}(\bC^{\tilde{N}})$ such that $\rho(\mathfrak{g})$ is a set containing skew-Hermitian matrices, we define $C(\rho) \in \bR$ such that \beq \Tr[\rho(E^\alpha)\rho(E^\beta)] = C(\rho)\Tr[E^\alpha E^\beta]. \label{e.ck.2} \eeq

Also define \beq \mathscr{E}(\rho) := -\sum_{\alpha=1}^N \rho(E^\alpha)\rho(E^\alpha) \nonumber \eeq to be its (quadratic) Casimir operator.
\end{notation}


\begin{defn}\label{d.y.1}
Refer to Definitions \ref{d.rat.1} and \ref{d.r.1} and Equation (\ref{e.b.1}). Let $\sigma: I_\delta^2 \rightarrow \bR^4$ be any parametrization of $R_\delta[a,T]$. In terms of this parametrization, we can write
\begin{align}
c\int_{R_\delta[a,T]} \fd A =& c\sum_{\alpha=1}^N\int_{I_\delta^2}d\hat{s}\ \sum_{j=1}^3 |J_{0j}^{\sigma}|(\hat{s}) [\fd_0 A_{j,\alpha}](\sigma(\hat{s}))\otimes E^\alpha \nonumber \\
=&c\sum_{\alpha=1}^N\int_{I_\delta^2}d\hat{s}\ \sum_{j=1}^3 |J_{0j}^{\sigma}|(\hat{s}) \left(\fd A,\tilde{\xi}_{0j}(\sigma(\hat{s}))\otimes E^\alpha\right)_\sharp \otimes E^\alpha \label{e.j.1} \\
=& \left(\fd A, c\sum_{\alpha=1}^N\int_{I_\delta^2}d\hat{s}\ \sum_{j=1}^3 |J_{0j}^{\sigma}|(\hat{s}) \tilde{\xi}_{0j}(\sigma(\hat{s}))\otimes E^\alpha\right)_\sharp \otimes E^\alpha. \nonumber
\end{align}
\end{defn}

Besides the renormalization rule mentioned in Remark \ref{r.rem.1} and adding a renormalization factor $\psi_w$ as discussed in Remark \ref{r.rem.2}, we also need a renormalization transformation as discussed in \cite{YMLim01}. We embed $\bR^4$ inside $\bC^4$, scaled by a factor $\kappa/2$. Thus, a surface $R_\delta[a,T] \subset \bR^4$ will be embedded inside $\bC^4$, but scaled with a factor $\kappa/2$.

Let $c = 1/\kappa$. The idea that the coupling constant goes to 0 as $\kappa$ goes to infinity, is known as asymptotic freedom. See \cite{YMLim01} and \cite{Gross01}.

\begin{rem}
Because of renormalization, note that we have to make the following changes to the RHS of Equation (\ref{e.j.1}).
\begin{itemize}
\item By introducing the renormalization factor $\psi_w$, we replace $\tilde{\xi}_{ab}$ in RHS of Equation (\ref{e.j.1}) with $\xi_{ab} = \psi\cdot \tilde{\xi}_{ab}$ .
\item By the renormalization rule $\fd_a \mapsto \kappa \fd_a$, we have to add an extra $\kappa$ in front of $\xi_{0i}$ in RHS of Equation (\ref{e.j.1}).
\item Because of the renormalization transformation, we have to scale $\sigma(s,t)$by a factor $\kappa/2$. Hence we have to replace \beq \sigma(s,t) \mapsto \kappa \sigma(s,t)/2. \nonumber \eeq As a consequence, we need to add an extra $\kappa^2/4$ to the surface integral in RHS of Equation (\ref{e.j.1}).
\end{itemize}
\end{rem}

\begin{defn}\label{d.p.2}
Suppose $\rho: \mathfrak{g} \rightarrow {\rm End}(\bC^{\tilde{N}})$ is an irreducible representation of a semi-simple Lie Algebra $\mathfrak{g}$, and $\rho(\mathfrak{g})$ consists of skew-Hermitian matrices.

Replace the coupling constant $c$ with $1/\kappa$. After applying the renormalization rule, transformation and adding a factor to Equation (\ref{e.j.1}), define
\begin{align*}
\nu_{R_\delta[a,T]}^{\kappa,\alpha} &:= \frac{1}{\kappa}\int_{I_\delta^2}d\hat{s}\ \frac{\kappa^2}{4}\sum_{j=1}^3 |J_{0j}^{\sigma}|(\hat{s}) \kappa\xi_{0j}(\kappa \sigma(\hat{s})/2)\otimes E^\alpha, \\
\left(\cdot, \nu_{R_\delta[a,T]}^{\kappa,\rho}\right)_\sharp &:= \sum_{\alpha=1}^N \left(\cdot, \nu_{R_\delta[a,T]}^{\kappa,\alpha}\right)_\sharp  \otimes \rho(E^\alpha),\quad \left(\cdot, \nu_{R_\delta[a,T]}^{\kappa} \right)_\sharp := \sum_{\alpha=1}^N \left(\cdot, \nu_{R_\delta[a,T]}^{\kappa,\alpha} \right)_\sharp \otimes E^\alpha,
\end{align*}
and $\mathcal{J}_{R_\delta[a,T]}^{\kappa,\rho} := \exp\left[\left(\cdot, \nu_{R_\delta[a,T]}^{\kappa,\rho} \right)_\sharp \right]$.
\end{defn}

\begin{rem}
By abuse of notation, for $B \in \dB \otimes \mathfrak{g} \subset B(\bC^4) \otimes \Lambda^2(\bR^4) \otimes \mathfrak{g}$,
\beq \left(B, \nu_{R_\delta[a,T]}^{\kappa,\rho} \right)_\sharp := \frac{1}{\kappa} \sum_{\alpha=1}^N\int_{I_\delta^2}d\hat{s}\ \frac{\kappa^2}{4}\sum_{j=1}^3 |J_{0j}^{\sigma}|(\hat{s}) \left(B, \kappa\xi_{0j}(\kappa\sigma(\hat{s})/2)\otimes E^\alpha\right)_\sharp \otimes \rho(E^\alpha), \nonumber \eeq so $\left(\cdot, \nu_{R_\delta[a,T]}^{\kappa,\rho}\right)_\sharp$ and $\left(\cdot, \nu_{R_\delta[a,T]}^{\kappa}\right)_\sharp$ are actually $\rho(\mathfrak{g})$-valued and $\mathfrak{g}$-valued random variables respectively.
Compare with Equation (\ref{e.j.1}).
\end{rem}

\begin{defn}\label{d.ym.2}(Definition of the Yang-Mills Path integral.)\\
Let $R[a,T]$ be a compact rectangular surface described in Definition \ref{d.rat.1}. Let $\rho$ be an irreducible representation of $\mathfrak{g}$. Consider the probability space $\dB \otimes \mathfrak{g}$ equipped with product Wiener measure $\tilde{\mu}_{\kappa^2}^{\times^{N}}$.

From Expression \ref{e.ym.1}, we define a sequence of Yang-Mills Path integrals using the renormalization flow \beq \left\{\Psi_\kappa: \mathcal{S}_\kappa(\bR^4)\otimes \Lambda^1(\bR^3) \otimes \mathfrak{g} \longrightarrow  H^2(\bC^4)\otimes \Lambda^1(\bR^3) \otimes \mathfrak{g}  \ |\ \kappa>0 \right\}, \nonumber \eeq
as
\begin{align*}
\left\{ \bE_{{\rm YM}}^\kappa\left[ \exp\left[\left(\cdot, \nu_{R[a,T]}^{\kappa,\rho}\right)_\sharp\right]\right] := \frac{1}{\mathbb{E} \left[ \mathcal{Y}^\kappa\right]}\mathbb{E}\left[ \mathcal{J}_{R[a,T]}^{\kappa,\rho} \cdot\mathcal{Y}^\kappa\right]  :\ \kappa > 0 \right\},
\end{align*}
where $\mathcal{Y}^\kappa$ was defined in Definition \ref{d.ym.3}. This is finite follows from applying Holder's Inequality for some $p$ in Corollary \ref{c.p.1}.

We will also write \beq \bE_{{\rm YM}}^\kappa\left[\Tr\ \exp\left[\left(\cdot, \nu_{R[a,T]}^{\kappa,\rho}\right)_\sharp \right]\right] \equiv \frac{1}{\mathbb{E} \left[ \mathcal{Y}^\kappa\right]}\mathbb{E}\left[ \Tr\ \exp\left[\left(\cdot,\nu_{R[a,T]}^{\kappa,\rho}\right)_\sharp \right] \cdot \mathcal{Y}^\kappa\right]. \nonumber \eeq

Here, it is understood that the expectation is taken on the probability space $(\dB \otimes \mathfrak{g},\tilde{\mu}_{\kappa^2}^{\times^{N}})$.
\end{defn}

\section{Wilson Area Law Formula}\label{s.cym}

We will now show how one can obtain the Area Law Formula for a compact semi-simple gauge group, using our definition of the Yang-Mills Path integral, given in Definition \ref{d.ym.2}.

On a finite lattice gauge with spacing $\epsilon$, one will define a finite dimensional integral, dependent on $\epsilon$, synonymous with $1/\kappa$. The Yang-Mills path integral is then defined by taking the limit as $\epsilon$ goes down to 0. The reader may refer to \cite{glimm1981quantum}, or other references on lattice gauge approximation in \cite{YMLim01}. We constructed a sequence of functionals $\{\mathcal{Y}^\kappa\}_{\kappa > 0}$ and $\{\mathcal{J}_{R[a,T]}^{\kappa,\rho}\}_{\kappa > 0}$, each defined on a Wiener space $\dB \otimes \mathfrak{g}$, equipped with product Wiener measure $\tilde{\mu}_{\kappa^2}^{\times^{N}}$ with variance $1/\kappa^2$. For each $\kappa > 0$, the path integral defined in Definition \ref{d.ym.2} is analogous to a finite dimensional integral on a finite lattice gauge with spacing $\epsilon$. It remains to take $\kappa$ going to infinity.

On a lattice gauge, Gross in \cite{Gross01} explained that the coupling constant $c(\epsilon)$ should depend on $\epsilon$, and as $\epsilon$ goes down to 0, $c(\epsilon)$ should go down to 0. This idea is also known as asymptotic freedom. When the gauge group is non-abelian, asymptotic freedom applies. See \cite{peskin1995introduction}. In Definition \ref{d.p.2}, we set the coupling constant $c = 1/\kappa$. This is crucial as it allows us to obtain the Area Law formula. (See Remark \ref{r.af.1}.) Without asymptotic freedom, the Area Law formula does not hold. Refer to \cite{YMLim01} for the case of abelian gauge group.

When $\xi_{ab}(w)$ and $\zeta_i(w)$ are complex-valued random variables with mean zero on $\dB$, with probability measure $\tilde{\mu}_{\kappa^2}$, we have \beq \left|\left(\cdot, \frac{\kappa\xi_{ab}(w)}{\parallel \xi_{ab}(w)\parallel}\right)_\sharp \right|^2,\quad \left|\left(\cdot, \frac{\kappa\zeta_i(w)}{\parallel \zeta_i(w) \parallel}\right)_\sharp \right|^2, \nonumber \eeq both are equal to $|Z_1 + \sqrt{-1}Z_2|^2$ in distribution, each $Z_i$ is a (real) normal random variable, with $\bE|Z_1|^2 + \bE|Z_2|^2 = 1$. See Equation (\ref{e.gs.1}).

Both $\parallel \xi_{ab}(w) \parallel$ and $\parallel \zeta_i(w) \parallel$ are less than or equal to $\sqrt{40/2\pi}$, from Remark \ref{r.rem.2}. Hence, we see that
\beq \bE[|Y_1^{\kappa}|^{2l}] = O(1/\kappa^{4l}), \quad \bE[|Y_2^{\kappa}|^{2l}] = O(1/\kappa^{4l}), \quad \bE[|Y_3^{\kappa}|^{n}] = O(1/\kappa^{4n}), \label{e.bds.1} \eeq and when $i=1,2$, $\bE[Y_i^{\kappa, 2l+1}] = 0$, for $l, n \in \mathbb{N}$. This is because the expectation, of an odd product of Gaussian random variables with mean 0, is zero.

Refer to Definition \ref{d.rat.1}. Given any $\mathbf{x}$ in the interior of $R_\delta[a,T]$, we proved in \cite{YMLim01} that
\begin{align}
\left(\frac{\kappa}{4}\right)^2\int_{I_\delta^2}d\hat{s}\ \sum_{j=1}^3 |J_{0j}^{\sigma}|(\hat{s})| \left\langle \xi_{0j}(\kappa \sigma(\hat{s})/2), \xi_{0j}(\kappa \mathbf{x}/2) \right\rangle
=& \frac{1}{4}\sum_{j=1}^3 \rho_\sigma^{0j}(\hat{t}) + O(e^{-\kappa^2\tilde{C}(\mathbf{x})}), \label{e.a.7}
\end{align}
whereby $\mathbf{x} = \sigma(\hat{t})$ and we can choose $\tilde{C}(\mathbf{x})$ to be $1/16$ of the shortest distance between $\mathbf{x}$ and the boundary of $R_\delta[a,T]$. See Definition \ref{d.r.1}.

Since for any $\mathbf{x} \in R[a,T] \equiv R_0[a,T]$, its distance to the boundary of $R_\delta[a,T]$ is at least the minimum of $\delta T$ and $\delta |a|$, we have
\begin{align}
\int_{I^2}d\hat{t}&\left(\frac{\kappa}{4}\right)^2\int_{I_\delta^2}d\hat{s}\ \sum_{j=1}^3 |J_{0j}^{\sigma}|(\hat{s})|J_{0j}^{\sigma}|(\hat{t}) \left\langle \xi_{0j}(\kappa \sigma(\hat{s})/2), \xi_{0j}(\kappa \sigma(\hat{t})/2) \right\rangle  \nonumber \\
&= \frac{1}{4}\int_{I^2}d\hat{t}\sum_{j=1}^3 \rho_\sigma^{0j}(\hat{t})|J_{0j}^{\sigma}|(\hat{t}) +  O\left(e^{-\kappa^2 C}\right) = \frac{|a|T}{4} + O\left(e^{-\kappa^2 C}\right), \label{e.a.5}
\end{align}
for $C = [{\rm min}\{T\delta, |a|\delta\}]^2/16 > 0$.

Furthermore, by using Dominated Convergence Theorem, Equation (\ref{e.a.7}) implies that
\begin{align}
\bE&\left[-\left( \cdot, \nu_{R_\delta[a,T]}^{\kappa,\rho} \right)_\sharp^2 \right]  = -\sum_{\alpha=1}^N \bE\left[\left( \cdot, \nu_{R_\delta[a,T]}^{\kappa,\alpha} \right)_\sharp^2 \right]\otimes \rho(E^\alpha)\rho(E^\alpha) \nonumber \\
&= -\sum_{\alpha=1}^N \frac{1}{\kappa^2}\left\langle  \nu_{R_\delta[a,T]}^{\kappa,\alpha}, \nu_{R_\delta[a,T]}^{\kappa,\alpha} \right\rangle\otimes \rho(E^\alpha)\rho(E^\alpha) \nonumber \\
&= \int_{I_\delta^2}d\hat{t}\left(\frac{\kappa}{4}\right)^2\int_{I_\delta^2}d\hat{s}\ \sum_{j=1}^3 |J_{0j}^{\sigma}|(\hat{s})|J_{0j}^{\sigma}|(\hat{t}) \left\langle \xi_{0j}(\kappa \sigma(\hat{s})/2), \xi_{0j}(\kappa \sigma(\hat{t})/2) \right\rangle\otimes \mathscr{E} \nonumber\\
&\longrightarrow \frac{1}{4}\int_{I_\delta^2}d\hat{t}\sum_{j=1}^3\rho_\sigma^{0j}(\hat{t})|J_{0j}^{\sigma}|(\hat{t}) \otimes \mathscr{E} \equiv  \frac{|a|T(1+2\delta)^2}{4}\otimes \mathscr{E}, \label{e.a.8}
\end{align}
as $\kappa \rightarrow \infty$.

\begin{prop}\label{p.s.2}
We have that $\{Y_3^1 > 0\}$ with probability 1, on the Wiener space $(\dB \otimes \mathfrak{g}, \tilde{\mu}_1^{\times^N})$.
\end{prop}

\begin{proof}
Since $\mathfrak{g}$ is semi-simple, without any loss of generality, $c^{1,\beta_0}_\gamma = -c_\gamma^{\beta_0,1}\neq 0$ for some $\beta_0, \gamma \neq 1$. By Definition \ref{d.ym.3}, it suffices to show that for any $1 \leq i < j \leq 3$, we have that $\sum_{\alpha,\beta=1}^N c_\gamma^{\alpha\beta}(\cdot, \pi_{i,\alpha}(w))_\sharp (\cdot, \pi_{j,\beta}(w))_\sharp \neq 0$ with probability 1, for any $w \in \bC^4$.

Now, $\{(\cdot, \pi_{i,\alpha}(w))_\sharp, (\cdot, \pi_{j,\beta}(w))_\sharp:\ 1 \leq \alpha,\beta \leq N\}$ is equal in distribution to $\{\tilde{N}_\alpha, N_\beta: 1 \leq \alpha,\beta \leq N\}$, which is a set containing $2N$ independent, identically distributed (i.i.d.), complex Gaussian variables, of mean zero. Hence, it suffices to show that $\sum_{\alpha,\beta=1}^N c_\gamma^{\alpha,\beta}\tilde{N}_\alpha N_\beta = 0$ with probability 0.

We have
\begin{align*}
\sum_{\alpha,\beta=1}^N c_\gamma^{\alpha,\beta}\tilde{N}_\alpha N_\beta =&  \sum_{\beta=2}^N c_\gamma^{1,\beta}\tilde{N}_1 N_\beta + \sum_{\alpha=2}^N \sum_{\beta=1}^N c_\gamma^{\alpha,\beta}\tilde{N}_\alpha N_\beta.
\end{align*}
For some $\beta_0$, $c_\gamma^{1,\beta_0} \neq 0$. Therefore, we see that $\sum_{\beta=2}^N c_\gamma^{1,\beta} N_\beta$ is a non-trivial linear combination of the complex Gaussian variables $\{N_\beta:\ 1 \leq \beta \leq N\}$, thus with probability 1, $\sum_{\beta=2}^N c_\gamma^{1,\beta} N_\beta \neq 0$. Hence, we can solve $\sum_{\alpha,\beta=1}^N c_\gamma^{\alpha,\beta}\tilde{N}_\alpha N_\beta = 0$ uniquely for $\tilde{N}_1$, \beq \tilde{N}_1 = -\frac{\sum_{\alpha=2}^N \tilde{N}_\alpha\sum_{\beta=1}^N c_\gamma^{\alpha,\beta} N_\beta}{\sum_{\beta=2}^N c_\gamma^{1,\beta}N_\beta}, \nonumber \eeq showing that $\sum_{\alpha,\beta=1}^N c_\gamma^{\alpha,\beta}\tilde{N}_\alpha N_\beta = 0$ with probability 0, since $\{\tilde{N}_\alpha: 1 \leq \alpha \leq N\}$ is a set of i.i.d. complex Gaussian random variables.
\end{proof}

\begin{lem}\label{l.bds.1}
There exists a $\kappa_0 > 0$, independent of $\rho$, such that for all $\kappa > \kappa_0$, we have positive constants $c_1, c_2 > 0$ independent of $\rho$, such that the trace $\Tr$, \beq c_1\frac{C(\rho)}{\kappa^4} \leq \Tr\ \bE\left[\left( \cdot, \nu_{R_\delta[a,T]}^{\kappa,\rho} \right)_\sharp^2(\mathcal{Y}^\kappa-1)\right] \leq c_2\frac{C(\rho)}{\kappa^4}. \nonumber \eeq Refer to Notation \ref{n.co.1}. Here, the expectation is taken with respect to the Wiener space $\dB \otimes \mathfrak{g}$ equipped with product Wiener measure $\tilde{\mu}_{\kappa^2}^{\times^{N}}$.
\end{lem}

\begin{proof}
Note that because $\rho(\mathfrak{g})$ consists of skew-Hermitian matrices, the square $\left( B, \nu_{R_\delta[a,T]}^{\kappa,\rho} \right)_\sharp^2$ is non-positive definite.

Refer to Definition \ref{d.ym.3}. From Equation (\ref{e.bds.1}), the important terms we need to consider are $Y_3^\kappa \geq 0$ and the squares of $Y_1^\kappa$ and $Y_2^\kappa$. First, we claim that \beq \Tr\ \bE\left[-\left( \cdot, \nu_{R_\delta[a,T]}^{\kappa,\rho} \right)_\sharp^2 Y_3^\kappa\right] \geq \frac{C(\rho)}{\kappa^4}c ,
\label{e.b.3} \eeq
for some constant $c > 0$, independent of $\kappa$ and $\rho$.

Now, $\left(\kappa^4 Y_3^\kappa, \tilde{\mu}_{\kappa^2}^{\times^{N}} \right)$ is equal in distribution to
$\left( Y_3^1, \tilde{\mu}_{1}^{\times^{N}} \right)$. And for any $\epsilon > 0$, we can find a $\tilde{c}(\epsilon) > 0$ such that \beq \bE\left[1_{\{\kappa^4 Y_3^\kappa > \epsilon\}} \right] = \bE\left[1_{\{ Y_3^1 > \epsilon\}} \right] = 1-\tilde{c}(\epsilon), \nonumber \eeq independent of $\kappa$, and converges to 1 as $\epsilon \rightarrow 0^+$. See Proposition \ref{p.s.2}.

For any $\alpha=1, \cdots, N$, $\left(\left( \cdot, \nu_{R_\delta[a,T]}^{\kappa,\alpha} \right)_\sharp^2 , \tilde{\mu}_{\kappa^2}^{\times^{N}} \right)$ is equal in distribution to \\
$\left( \left( \cdot, \frac{1}{\kappa}\nu_{R_\delta[a,T]}^{\kappa,\alpha} \right)_\sharp^2 , \tilde{\mu}_{1}^{\times^{N}} \right)$. Apply Cauchy Schwartz inequality and Equation (\ref{e.ck.2}),
\begin{align*}
\Tr\ &\bE\left[- \left( \cdot, \frac{1}{\kappa}\nu_{R_\delta[a,T]}^{\kappa,\rho} \right)_\sharp^2 1_{\{Y_3^1 \leq \epsilon\}}\right] \leq NC(\rho)\left|\bE \left[ \left( \cdot, \frac{1}{\kappa}\nu_{R_\delta[a,T]}^{\kappa,1} \right)_\sharp^4\right]\right|^{1/2} \left|\bE\left[1_{\{Y_3^1 \leq \epsilon\}} \right]\right|^{1/2} \\
=&  \frac{\sqrt3NC(\rho)}{\kappa^2}\left\langle \nu_{R_\delta[a,T]}^{\kappa,1}, \nu_{R_\delta[a,T]}^{\kappa,1}\right\rangle \left|\bE\left[1_{\{Y_3^1 \leq \epsilon\}} \right]\right|^{1/2}
= \sqrt3\Tr\ \bE\left[- \left( \cdot, \nu_{R_\delta[a,T]}^{\kappa,\rho} \right)_\sharp^2 \right]\sqrt{\tilde{c}(\epsilon)},
\end{align*}
which goes to 0 as $\epsilon \rightarrow 0^+$. In the last equality, the expectation is taken with respect to $\tilde{\mu}_{\kappa^2}^{\times^N}$.

Hence, we have for $\epsilon$ small enough,
\begin{align*}
\Tr\ \bE\left[- \left( \cdot, \frac{1}{\kappa}\nu_{R_\delta[a,T]}^{\kappa,\rho} \right)_\sharp^2 1_{\{Y_3^1 > \epsilon\}}\right] &>\frac{1}{2} \Tr\ \bE\left[- \left( \cdot, \nu_{R_\delta[a,T]}^{\kappa,\rho} \right)_\sharp^2 \right] \\
&= \frac{NC(\rho)}{2} \bE\left[\left( \cdot, \nu_{R_\delta[a,T]}^{\kappa,1} \right)_\sharp^2 \right],
\end{align*}
the lower bound holds for any $\kappa$.

Together with Equation (\ref{e.a.8}), this will imply that
\begin{align*}
\Tr\ &\bE\left[- \left( \cdot, \nu_{R_\delta[a,T]}^{\kappa,\rho} \right)_\sharp^2  Y_3^\kappa \right]
= \frac{1}{\kappa^4}\Tr\ \bE\left[- \left( \cdot, \nu_{R_\delta[a,T]}^{\kappa,\rho} \right)_\sharp^2  \kappa^4 Y_3^\kappa \right] \\
\geq& \frac{1}{\kappa^4}\Tr\ \bE\left[- \left( \cdot, \nu_{R_\delta[a,T]}^{\kappa,\rho} \right)_\sharp^2  \kappa^4 Y_3^\kappa 1_{\{\kappa^4 Y_3^\kappa > \epsilon\}}\right] =
\frac{1}{\kappa^4}\Tr\ \bE\left[- \left( \cdot, \frac{1}{\kappa}\nu_{R_\delta[a,T]}^{\kappa,\rho} \right)_\sharp^2  Y_3^1 1_{\{Y_3^1 > \epsilon\}}\right]\\
\geq& \epsilon \frac{1}{\kappa^4}\Tr\ \bE\left[- \left( \cdot, \frac{1}{\kappa}\nu_{R_\delta[a,T]}^{\kappa,\rho} \right)_\sharp^2 1_{\{ Y_3^1 > \epsilon\}}\right] \geq \epsilon \frac{1}{\kappa^4}\tilde{c} C(\rho),
\end{align*}
for some constant $\tilde{c}>0$ independent of $\kappa$ and $\rho$. This proves Equation (\ref{e.b.3}).

The complex random variable $(\cdot, \pi_{i,\alpha}(w))_\sharp$
is independent of $(\cdot, \pi_{i,\beta}(w))_\sharp$, if $\alpha \neq \beta$. And
\beq \bE \left[ (\cdot, \pi_{i,\alpha}(w))_\sharp \overline{(\cdot, \pi_{j,\alpha}(w))_\sharp} \right] = \bE \left[ (\cdot, \pi_{i,\alpha}(w))_\sharp (\cdot, \pi_{j,\alpha}(\bar{w}))_\sharp \right] = \frac{1}{\kappa^2}\langle \pi_{i,\alpha}(w), \pi_{j,\alpha}(w) \rangle, \nonumber \eeq so $(\cdot, \pi_{i,\alpha}(w))_\sharp$
has variance $O(1/\kappa^2)$. Note that $\langle \cdot, \cdot \rangle$ is a sesquilinear complex inner product. Thus, using Cauchy Schwartz inequality and Equation (\ref{e.a.8}),
\begin{align}
&\left|\Tr\ \bE\left[\left( \cdot, \nu_{R_\delta[a,T]}^{\kappa,\rho} \right)_\sharp^2 \left|\sum_{\alpha , \beta}c_\gamma^{\alpha\beta}  (\cdot, \pi_{i,\alpha}(w))_\sharp (\cdot, \pi_{j,\beta}(w))_\sharp \right|^2  \right]\right|\nonumber \\
&= -\Tr\ \bE\left[\left( \cdot, \nu_{R_\delta[a,T]}^{\kappa,\rho} \right)_\sharp^2 \left|\sum_{\alpha , \beta}c_\gamma^{\alpha\beta}  (\cdot, \pi_{i,\alpha}(w))_\sharp (\cdot, \pi_{j,\beta}(w))_\sharp \right|^2  \right]\nonumber\\
&= C(\rho)\sum_{\alpha=1}^N\bE\left[\left( \cdot, \nu_{R_\delta[a,T]}^{\kappa,\alpha} \right)_\sharp^2 \left|\sum_{\alpha , \beta}c_\gamma^{\alpha\beta}  (\cdot, \pi_{i,\alpha}(w))_\sharp (\cdot, \pi_{j,\beta}(w))_\sharp \right|^2  \right]\nonumber \\
&\leq C(\rho)\sum_{\alpha=1}^N \left[\bE\left( \cdot, \nu_{R_\delta[a,T]}^{\kappa,\alpha} \right)_\sharp^4 \right]^{1/2} \left[\bE\left|\sum_{\alpha , \beta}c_\gamma^{\alpha\beta}  (\cdot, \pi_{i,\alpha}(w))_\sharp (\cdot, \pi_{j,\beta}(w))_\sharp \right|^4\right]^{1/2} \nonumber \\
&\leq C(\rho)\underline{C}\frac{1}{\kappa^4},\label{e.c.5}
\end{align}
for some constant $\underline{C}>0$ independent of $\rho$ and $\kappa$, from Equation (\ref{e.a.8}).

On the probability space $(\dB \otimes \mathfrak{g}, \tilde{\mu}_{\kappa^2}^{\times^{N}})$, we have
\begin{align*}
\int_{\dB \otimes \mathfrak{g}} \left( \cdot, \nu_{R_\delta[a,T]}^{\kappa,\rho} \right)_\sharp^2 Y_3^\kappa\ d\tilde{\mu}_{\kappa^2}^{\times^{N}} \equiv& \bE\left[\left( \cdot, \nu_{R_\delta[a,T]}^{\kappa,\rho} \right)_\sharp^2Y_3^\kappa\right]\\
=& \sum_{\gamma=1}^N \sum_{1\leq i<j \leq 3}\int_{w\in \bC^4}d\lambda_4\ C_{ij}^\gamma(w),
\end{align*}
whereby
\begin{align*}
C_{ij}^\gamma(w) :=& \bE\left[\left( \cdot, \nu_{R_\delta[a,T]}^{\kappa,\rho} \right)_\sharp^2 \left|\sum_{\alpha , \beta}c_\gamma^{\alpha\beta}  (\cdot, \pi_{i,\alpha}(w))_\sharp (\cdot, \pi_{j,\beta}(w))_\sharp \right|^2  \right].
\end{align*}
From the first claim, we see that $-\Tr\ C_{ij}^\gamma(w) \equiv \Tr\ |C_{ij}^\gamma(w)| > 0$ for some $\gamma$ and from Equation (\ref{e.c.5}), \beq \left|\Tr\ \bE\left[\left( \cdot, \nu_{R_\delta[a,T]}^{\kappa,\rho} \right)_\sharp^2Y_3^\kappa\right] \right| \leq 3\underline{C}NC(\rho)\frac{1}{\kappa^4}. \label{e.c.6} \eeq

Now, if the structure constant $c_\gamma^{\alpha\beta}$ is non-zero, then all the $\alpha, \beta$ and $\gamma$ must be distinct. And $(\cdot, \xi_{ij}(w))_\sharp$ is independent of $(\cdot, \xi_{\hat{i}\hat{j}}(w))_\sharp$, provided $(i,j) \neq (\hat{i},\hat{j})$, with \beq \bE\left[\left(\cdot, \xi_{ij}(w) \right)_\sharp  \overline{\left(\cdot, \xi_{ij}(w) \right)_\sharp} \right] = \frac{1}{\kappa^2}\parallel \xi_{ij}(w) \parallel^2 = \frac{1}{2\pi \kappa^2}, \nonumber \eeq
on the probability space $(\dB ,\tilde{\mu}_{\kappa^2})$.

Write \beq \widetilde{\sum} = \sum_{\gamma=1}^N \sum_{1\leq i < j \leq 3}, \quad \nu^\kappa = \nu_{R_\delta[a,T]}^{\kappa,\rho}. \nonumber \eeq Thus,
\begin{align*}
\bE&\left[ \left(\widetilde{\sum}\left(\cdot, \xi_{ij}(w)\otimes E^\gamma \right)_\sharp \right)^2\right] \\
&= \frac{1}{\kappa^2}\sum_{\gamma=1}^N \sum_{1\leq i < j \leq 3}\sum_{\hat{\gamma}=1}^N \sum_{1\leq \hat{i} < \hat{j} \leq 3}\left\langle \xi_{ij}(w)\otimes E^\gamma, \xi_{\hat{i}\hat{j}}(w)\otimes E^{\hat{\gamma}} \right\rangle= \widetilde{\sum} \frac{1}{\kappa^2}\parallel \xi_{ij}(w)\parallel^2.
\end{align*}

Now, $\{(\cdot, \xi_{ij,\gamma}(w) )_\sharp: 1\leq i < j \leq 3,\ 1 \leq \gamma \leq N\}$ is a set of i.i.d. complex random variables, independent of
\begin{itemize}
  \item $\left( \cdot, \nu_{R_\delta[a,T]}^\kappa \right)_\sharp$, because they are independent of $(\cdot, \xi_{0j,\gamma}(w) )_\sharp$, for $j=1,2,3$;
  \item $(\cdot, \pi_{i,\alpha}(w))_\sharp$, because $\zeta_i(w) = \zeta(w) \otimes dx^0 \wedge dx^i$ and when $c_\gamma^{\alpha\beta}$ is non-zero.
\end{itemize}

Therefore,
\begin{align*}
\bE&\left[\left|\left( \cdot, \nu^\kappa \right)_\sharp Y_2^\kappa \right|^2 \right] = \bE\left[-\left( \cdot, \nu^\kappa \right)_\sharp^2 \left|Y_2^\kappa \right|^2 \right] \\
=& \bE\left[\left|\left( \cdot, \nu^\kappa \right)_\sharp\widetilde{\sum}\sum_{\alpha , \beta}c_\gamma^{\alpha\beta}\int_{w \in \bC^4}d\lambda_4(w)\
(\cdot, \pi_{i,\alpha}(w))_\sharp (\cdot, \pi_{j,\beta}(w))_\sharp
(\cdot, \kappa \xi_{ij,\gamma}(\bar{w}))_\sharp  \right|^2 \right] \\
=& \widetilde{\sum}\bE\left[\left|\left( \cdot, \nu^\kappa \right)_\sharp\int_{w \in \bC^4}d\lambda_4(w)\ \sum_{\alpha , \beta}c_\gamma^{\alpha\beta}  (\cdot, \pi_{i,\alpha}(w))_\sharp (\cdot, \pi_{j,\beta}(w))_\sharp (\cdot, \kappa \xi_{ij,\gamma}(\bar{w}))_\sharp \right|^2\right].
\end{align*}

By Jensen's inequality and independence of the random variables, we have for $\Xi_{ij} := \bE\left[\left|(\cdot, \kappa\xi_{ij, \gamma}(w))_\sharp \right|^2\right]$,
\begin{align*}
\Tr\ &\bE\left[\left|\left( \cdot, \nu^\kappa \right)_\sharp Y_2^\kappa \right|^2 \right] =
\bE\left[-\Tr\ \left( \cdot, \nu^\kappa \right)_\sharp^2 \left|Y_2^\kappa \right|^2 \right] \\
\leq& \widetilde{\sum}\int_{w \in \bC^4}d\lambda_4(w)\ \bE\left[-\Tr\ \left( \cdot, \nu^\kappa \right)_\sharp^2\left|\sum_{\alpha , \beta}c_\gamma^{\alpha\beta}  (\cdot, \pi_{i,\alpha}(w))_\sharp (\cdot, \pi_{j,\beta}(w))_\sharp (\cdot, \kappa \xi_{ij,\gamma}(\bar{w}))_\sharp \right|^2 \right]\\
=&  \widetilde{\sum}\int_{\bC^4}d\lambda_4(w)\ \Xi_{ij} \cdot \bE\left[-\Tr\ \left( \cdot, \nu^\kappa \right)_\sharp^2\left|\sum_{\alpha , \beta}c_\gamma^{\alpha\beta}  (\cdot, \pi_{i,\alpha}(w))_\sharp (\cdot, \pi_{j,\beta}(w))_\sharp  \right|^2 \right] \\
=& \sum_{\gamma=1}^N \sum_{1\leq i < j \leq 3} \int_{w \in \bC^4}d\lambda_4(w)\ \parallel \xi_{ij}(w)\parallel^2 \Tr\ |C_{ij}^\gamma(w)|.
\end{align*}

Since \beq \parallel \xi_{ij}(w)\parallel^2  = \frac{1}{2\pi}, \nonumber \eeq we see that \beq \Tr\ \bE\left[\left|\left( \cdot, \nu_{R_\delta[a,T]}^{\kappa,\rho} \right)_\sharp Y_2^\kappa \right|^2\right] \leq \frac{1}{2\pi}\sum_{\gamma=1}^N \sum_{1\leq i<j \leq 3}\int_{w \in \bC^4}d\lambda_4(w)\ \Tr\ |C_{ij}^\gamma(w)|. \nonumber \eeq

Similarly,
\beq \Tr\ \bE\left[\left|\left( \cdot, \nu_{R_\delta[a,T]}^{\kappa,\rho} \right)_\sharp Y_1^\kappa \right|^2 \right] \leq \frac{1}{2\pi}\sum_{\gamma=1}^N \sum_{1\leq i<j \leq 3}\int_{w \in \bC^4}d\lambda_4(w)\ \Tr\ |C_{ij}^\gamma(w)|. \nonumber \eeq

By Jensen's Inequality,
\begin{align*}
\Tr\ &\bE\left[\left|\left( \cdot, \nu_{R_\delta[a,T]}^{\kappa,\rho} \right)_\sharp Y_1^\kappa + \left( \cdot, \nu_{R_\delta[a,T]}^{\kappa,\rho} \right)_\sharp Y_2^\kappa\right|^2 \right] \\
=& 4\bE \left[-\Tr\ \left( \cdot, \nu_{R_\delta[a,T]}^{\kappa,\rho} \right)_\sharp^2 \left| \frac{1}{2}(Y_1^\kappa + Y_2^\kappa) \right|^2\right]\\
\leq& 2\bE\left[\Tr\  \left|\left( \cdot, \nu_{R_\delta[a,T]}^{\kappa,\rho} \right)_\sharp^2\right||Y_1^\kappa|^2 + \Tr\ \left|\left( \cdot, \nu_{R_\delta[a,T]}^{\kappa,\rho} \right)_\sharp^2\right||Y_2^\kappa|^2 \right] \\
\leq& \frac{4}{2\pi}\sum_{\gamma=1}^N \sum_{1\leq i<j \leq 3}\int_{w \in \bC^4}d\lambda_4(w)\ \Tr\ |C_{ij}^\gamma(w)|.
\end{align*}

Therefore, we see that
\begin{align*}
-&\left(1 + \frac{1}{2\pi}\right)\Tr\ \bE\left[\left( \cdot, \nu_{R_\delta[a,T]}^{\kappa,\rho} \right)_\sharp^2 Y_3^\kappa\right] \\
&\geq\frac{1}{4}\Tr\ \bE\left[\left( \cdot, \nu_{R_\delta[a,T]}^{\kappa,\rho} \right)_\sharp^2(Y_1^\kappa + Y_2^\kappa)^2\right] - \Tr\ \bE\left[\left( \cdot, \nu_{R_\delta[a,T]}^{\kappa,\rho} \right)_\sharp^2 Y_3^\kappa\right] \\
&\geq - \left(1-\frac{1}{2\pi}\right)\Tr\ \bE\left[\left( \cdot, \nu_{R_\delta[a,T]}^{\kappa,\rho} \right)_\sharp^2 Y_3^\kappa\right] > 0.
\end{align*}
From Equations (\ref{e.b.3}) and (\ref{e.c.6}), there exist $\dot{c}, \hat{c} > 0$, independent of $\rho$, such that \beq \frac{\dot{c}C(\rho)}{\kappa^4} \leq \frac{1}{4}\Tr\ \bE\left[\left( \cdot, \nu_{R_\delta[a,T]}^{\kappa,\rho} \right)_\sharp^2(Y_1^\kappa + Y_2^\kappa)^2\right] - \Tr\ \bE\left[\left( \cdot, \nu_{R_\delta[a,T]}^{\kappa,\rho} \right)_\sharp^2 Y_3^\kappa\right] \leq \frac{\hat{c}C(\rho)}{\kappa^4} . \label{e.bds.3} \eeq

Define
\begin{align*}
F^\kappa :=& \mathcal{Y}^\kappa -1 + \frac{1}{2}\sum_{i=1}^3 Y_i^\kappa -\frac{1}{8}\left(\sum_{i=1}^3 Y_i^\kappa \right)^2  \\
= & \sum_{k=3}^\infty \left(-\frac{1}{2}\right)^k \frac{1}{k!}\left(\sum_{i=1}^3 Y_i^\kappa \right)^k.
\end{align*}

Corollary \ref{c.p.1} says that $\mathcal{Y}^\kappa$ is actually $L^p$ integrable, for any $1 < p < \pi$. For any $1 \leq \alpha \leq N$, we have \beq \bE\left[ \left(\cdot, \nu_{R_\delta[a,T]}^{\kappa,\alpha}\right)_\sharp^4 \right] = 3 \left[\frac{1}{\kappa^2}\left\langle \nu_{R_\delta[a,T]}^{\kappa,\alpha}, \nu_{R_\delta[a,T]}^{\kappa,\alpha} \right\rangle \right]^2 \longrightarrow 3\frac{|a|T(1+2\delta)^2}{4}, \label{e.a.1} \eeq as $\kappa \rightarrow \infty$. Refer to Equation (\ref{e.a.8}). Hence $\left|\left( \cdot, \nu_{R_\delta[a,T]}^{\kappa,\alpha} \right)_\sharp^2\right| F^\kappa$ is also $L^p$ integrable, for some $p > 1$.

From Equation (\ref{e.y.3}), we have \beq \sum_{\alpha=1}^N\bE\left[ \left|\left( \cdot, \nu_{R_\delta[a,T]}^{\kappa,\alpha} \right)_\sharp^2\right|^p |F^\kappa|^p\right] \leq C_p, \label{e.y.4} \eeq for some constant $C_p$ independent of $\kappa$.

Write $\mathbb{I}_M := 1_{\{\sum_{i=1}^3|Y_i^\kappa| \leq  M\}}$. For such value of $p$, let $q> 0$ such that $1/q + 1/p = 1$. For any positive number $M$, we have from Equation (\ref{e.ck.2}),
\begin{align}
\bE&\left|\Tr\ \left( \cdot, \nu_{R_\delta[a,T]}^{\kappa,\rho} \right)_\sharp^2 F^\kappa \right| \equiv \bE\left[ -\Tr\ \left( \cdot, \nu_{R_\delta[a,T]}^{\kappa,\rho} \right)_\sharp^2 \left|F^\kappa \right|\right]  \nonumber\\
=& \bE\left|1_{\{\sum_{i=1}^3|Y_i^\kappa| > M\}}\Tr\ \left( \cdot, \nu_{R_\delta[a,T]}^{\kappa,\rho} \right)_\sharp^2F^\kappa\right| + \bE\left|1_{\{\sum_{i=1}^3|Y_i^\kappa| \leq  M\}}\Tr\ \left( \cdot, \nu_{R_\delta[a,T]}^{\kappa,\rho} \right)_\sharp^2F^\kappa\right| \nonumber\\
=& C(\rho)\sum_{\alpha=1}^N\bE\left|1_{\{\sum_{i=1}^3|Y_i^\kappa| > M\}}\left( \cdot, \nu_{R_\delta[a,T]}^{\kappa,\alpha} \right)_\sharp^2F^\kappa\right| + \bE\left|1_{\{\sum_{i=1}^3|Y_i^\kappa| \leq  M\}}\Tr\ \left( \cdot, \nu_{R_\delta[a,T]}^{\kappa,\rho} \right)_\sharp^2F^\kappa\right| \nonumber\\
\leq& C(\rho)\left[\bE\left|1_{\{\sum_{i=1}^3|Y_i^\kappa| > M\}}\right|\right]^{1/q} \sum_{\alpha=1}^N\left(\bE\left[\left|\left( \cdot, \nu_{R_\delta[a,T]}^{\kappa,\alpha}\right)_\sharp^2 \right|^p |F^{\kappa}|^p\right] \right)^{1/p} \nonumber \\
&+  \bE\left|\mathbb{I}_M\Tr\ \left( \cdot, \nu_{R_\delta[a,T]}^{\kappa,\rho} \right)_\sharp^2F^\kappa\right|, \label{e.c.2}
\end{align}
after applying Holder's Inequality.

Write \beq \widehat{\nu}^\kappa(p) := \sum_{\alpha=1}^N\left(\bE\left[\left|\left( \cdot, \nu_{R_\delta[a,T]}^{\kappa,\alpha}\right)_\sharp^2 \right|^p |F^{\kappa}|^p\right] \right)^{1/p} \equiv N\left(\bE\left[\left|\left( \cdot, \nu_{R_\delta[a,T]}^{\kappa,1}\right)_\sharp^2\right|^p |F^{\kappa}|^p\right] \right)^{1/p}. \nonumber \eeq Apply Chebyshev's and Jensen's Inequalities, the RHS of Equation (\ref{e.c.2}) is less than or equal to
\begin{align*}
\frac{C(\rho)}{M^{2n/q}}&\left[\bE\left|\left(\sum_{i=1}^3|Y_i^\kappa|\right)^{2n}\right|\right]^{1/q}
\widehat{\nu}^\kappa(p)
+ \bE\left|\mathbb{I}_M\Tr\ \left( \cdot, \nu_{R_\delta[a,T]}^{\kappa,\rho} \right)_\sharp^2F^\kappa\right| \\
&\leq \frac{C(\rho)3^{2n-1}}{M^{2n/q}}\left[\bE\sum_{i=1}^3|Y_i^\kappa|^{2n} \right]^{1/q}
\widehat{\nu}^\kappa(p) + \bE\left|\mathbb{I}_M\Tr\ \left( \cdot, \nu_{R_\delta[a,T]}^{\kappa,\rho} \right)_\sharp^2F^\kappa\right| \\
&\leq C(\rho)\frac{\tilde{C}(n)}{M^{2n/q}\kappa^{4n/q}}\widehat{\nu}^\kappa(p) +  \bE\left|\mathbb{I}_M \Tr\ \left( \cdot, \nu_{R_\delta[a,T]}^{\kappa,\rho} \right)_\sharp^2F^\kappa\right|,
\end{align*}
for some constant $\tilde{C}(n)$, dependent only on $n \in \mathbb{N}$.

Choose $n$ such that $n/q \geq 2$. From Equation (\ref{e.y.4}), we see that \beq \bE\left|\Tr\ \left( \cdot, \nu_{R_\delta[a,T]}^{\kappa,\rho} \right)_\sharp^2 F^\kappa \right| = O(C(\rho)/M^{2n/q}\kappa^8)  + \bE\left|1_{\{\sum_{i=1}^3|Y_i^\kappa| \leq  M\}}\Tr\ \left( \cdot, \nu_{R_\delta[a,T]}^{\kappa,\rho} \right)_\sharp^2F^\kappa\right|. \label{e.bds.2} \eeq

And,
\begin{align*}
\bE&\left|1_{\{\sum_{i=1}^3|Y_i^\kappa| \leq  M\}}\Tr\ \left( \cdot, \nu_{R_\delta[a,T]}^{\kappa,\rho} \right)_\sharp^2F^\kappa\right| \\
&\leq \sum_{k=3}^\infty \frac{1}{2^kk!}\bE\left[ 1_{\{\sum_{i=1}^3|Y_i^\kappa| \leq  M\}}\left|\Tr\ \left( \cdot, \nu_{R_\delta[a,T]}^{\kappa,\rho} \right)_\sharp^2\right|\left( \sum_{i=1}^3|Y_i^\kappa| \right)^k \right] \\
&\leq \sum_{k=3}^\infty \frac{M^{k-3}}{2^k k!}\bE\left[\left|\Tr\  \left( \cdot, \nu_{R_\delta[a,T]}^{\kappa,\rho} \right)_\sharp^2\left( \sum_{i=1}^3|Y_i^\kappa| \right)^3 \right|\right] \\
&= \left(\sum_{k=0}^\infty \frac{M^{k}}{2^{k+3}(k+3)!}\right) \times \bE\left|\Tr\  \left( \cdot, \nu_{R_\delta[a,T]}^{\kappa,\rho} \right)_\sharp^2\left( \sum_{i=1}^3|Y_i^\kappa| \right)^3 \right|,
\end{align*}
so its trace is $O(e^{M/2} C(\rho)/\kappa^6)$, from Equations (\ref{e.ck.2}) and (\ref{e.bds.1}).

Using this bound in Equation (\ref{e.bds.2}), we see that
\beq \left|\bE\left[\Tr\ \left( \cdot, \nu_{R_\delta[a,T]}^{\kappa,\rho} \right)_\sharp^2F^\kappa\right]\right| = O(C(\rho)/M^{2n/q}\kappa^8)  +  O(e^{M/2} C(\rho)/\kappa^6). \nonumber \eeq

Together with the bounds in Equation (\ref{e.bds.1}), we see that for a given $M \geq 1$,
\begin{align}
\bE&\left[\left( \cdot, \nu_{R_\delta[a,T]}^{\kappa,\rho} \right)_\sharp^2 (\mathcal{Y}^\kappa -1 )\right] \nonumber\\
=& -\frac{1}{2}\bE\left[\left( \cdot, \nu_{R_\delta[a,T]}^{\kappa,\rho} \right)_\sharp^2 Y_3^\kappa\right] + \frac{1}{8}\bE\left[ \left( \cdot, \nu_{R_\delta[a,T]}^{\kappa,\rho} \right)_\sharp^2\left(\sum_{i=1}^2 Y_i^\kappa \right)^2\right] \nonumber\\
&+ \frac{1}{8}\bE\left[ \left( \cdot, \nu_{R_\delta[a,T]}^{\kappa,\rho} \right)_\sharp^2Y_3^{\kappa,2} + 2\left( \cdot, \nu_{R_\delta[a,T]}^{\kappa,\rho} \right)_\sharp^2 Y_3^\kappa\sum_{i=1}^2 Y_i^\kappa\right] + \bE\left[\left( \cdot, \nu_{R_\delta[a,T]}^{\kappa,\rho} \right)_\sharp^2 F^\kappa \right]\nonumber \\
=& -\frac{1}{2}\bE\left[\left( \cdot, \nu_{R_\delta[a,T]}^{\kappa,\rho} \right)_\sharp^2Y_3^\kappa\right] + \frac{1}{8}\bE\left[ \left( \cdot, \nu_{R_\delta[a,T]}^{\kappa,\rho} \right)_\sharp^2\left(\sum_{i=1}^2 Y_i^\kappa \right)^2\right] + \epsilon(\rho, \kappa), \label{e.a.9}
\end{align}
whereby $ |\Tr\ \epsilon(\rho, \kappa)| = O(e^{M/2} C(\rho)/\kappa^6) + O(C(\rho)/\kappa^6) + O(C(\rho)/M^{4}\kappa^8)$.

Thus, we see from Equation (\ref{e.bds.3}), that there exists a large $\kappa_0 > 1$ independent of $\rho$, such that for all $\kappa > \kappa_0$, we have \beq c_1\frac{C(\rho)}{\kappa^4} < \Tr\ \bE \left[\left( \cdot, \nu_{R_\delta[a,T]}^{\kappa, \rho} \right)_\sharp^2(\mathcal{Y}^\kappa -1) \right] < c_2\frac{C(\rho)}{\kappa^4} \nonumber \eeq for some constants $c_1, c_2 > 0$ independent of $\rho$. This completes the proof.
\end{proof}

\begin{rem}\label{r.bds.1}
Note that Equation (\ref{e.bds.3}) still holds true, if we replace the factor $1/\sqrt{2\pi}$ in $\psi_w$ with any $0<\tilde{c} < 1$. Also see Remark \ref{r.p.1}.
\end{rem}

\begin{cor}
There exists a $\delta > 0$ such that for some $\kappa_0$, which depends on $\delta$ but independent of $\rho$, we have for all $\kappa > \kappa_0$, \beq -\bE \left[
\left( \cdot, \nu_{R[a,T]}^{\kappa,\rho} \right)_\sharp \left( \cdot, \nu_{R_\delta[a,T]}^{\kappa,\rho} \right)_\sharp
\mathcal{Y}^\kappa\right]  = \frac{|a|T}{4} \otimes \mathscr{E}(\rho) - \epsilon(\rho, \kappa), \nonumber \eeq whereby \beq 0< c_3\frac{C(\rho)}{\kappa^4} \leq \Tr\  \epsilon(\rho,\kappa)  \leq c_4\frac{C(\rho)}{\kappa^4}, \nonumber \eeq for some positive constants $c_3, c_4$, both are independent of $\kappa$ and $\rho$.
\end{cor}

\begin{proof}
For $\Lambda = \sum_{\alpha=1}^N\rho(E^\alpha)\rho(E^\alpha)$, we see from Equation (\ref{e.a.5}) that
\begin{align}
\bE&\left[ \left(\cdot, \nu_{R[a,T]}^{\kappa,\rho} \right)_\sharp\left(\cdot, \nu_{R_\delta[a,T]}^{\kappa,\rho} \right)_\sharp\right] \nonumber\\
&=  \int_{I^2}d\hat{t}\left(\frac{\kappa}{4}\right)^2\int_{I_\delta^2}d\hat{s}\ \sum_{j=1}^3 |J_{0j}^{\sigma}|(\hat{s})|J_{0j}^{\sigma}|(\hat{t}) \left\langle \xi_{0j}(\kappa \sigma(\hat{s})/2), \xi_{0j}(\kappa \sigma(\hat{t})/2) \right\rangle
\otimes \Lambda\nonumber \\
&=
\frac{|a|T}{4}\sum_{\alpha=1}^N \rho(E^\alpha)\rho(E^\alpha) + \tilde{\epsilon}(\rho,\kappa). \label{e.a.6}
\end{align}
Here, $\tilde{\epsilon}(\rho,\kappa)$ is the matrix reminder term, of which its trace \\
$|\Tr\ \tilde{\epsilon}(\rho,\kappa)| = O\left(C(\rho)e^{-\kappa^2 \bar{C}}\right)$, for some $\bar{C} > 0$, depending on $\delta$.

Write
\begin{align*}
\lambda_\delta^{\kappa,\alpha} = \nu_{R_\delta[a,T]}^{\kappa,\alpha} - \nu_{R[a,T]}^{\kappa,\alpha}, \quad \lambda_\delta^\kappa = \nu_{R_\delta[a,T]}^{\kappa,\rho} - \nu_{R[a,T]}^{\kappa,\rho} \equiv \sum_{\alpha=1}^N \lambda_\delta^{\kappa,\alpha} \otimes \rho(E^\alpha).
\end{align*}

Now, $R_\delta[a,T] \setminus R[a,T]$ is a flat compact surface of area $4|a|T(\delta + \delta^2)$. By Cauchy Schwartz inequality, we have from Equations (\ref{e.ck.2}) and (\ref{e.a.8}),
\begin{align*}
& \left|\Tr\ \bE \left[
\left( \cdot, \nu_{R[a,T]}^{\kappa,\rho} \right)_\sharp \left( \cdot, \lambda_\delta^\kappa \right)_\sharp
Y_3^\kappa\right] \right| \leq C(\rho)\sum_{\alpha=1}^N \left| \bE \left[
\left( \cdot, \nu_{R[a,T]}^{\kappa,\alpha} \right)_\sharp \left( \cdot, \lambda_\delta^{\kappa,\alpha} \right)_\sharp
Y_3^\kappa\right] \right| \\
&\leq C(\rho)\sum_{\alpha=1}^N \left| \bE \left[
\left( \cdot, \nu_{R[a,T]}^{\kappa,\alpha} \right)_\sharp^2 \left( \cdot, \lambda_\delta^{\kappa,\alpha} \right)_\sharp^2 \right]\right|^{1/2} \left| \bE \left[|Y_3^\kappa|^2 \right] \right|^{1/2} \\
&= C(\rho)\sum_{\alpha=1}^N \left|\left\langle \nu_{R[a,T]}^{\kappa,\alpha}, \nu_{R[a,T]}^{\kappa,\alpha} \right\rangle \left\langle \lambda_\delta^{\kappa,\alpha}, \lambda_\delta^{\kappa,\alpha}\right\rangle + 2\left\langle \nu_{R[a,T]}^{\kappa,\alpha}, \lambda_\delta^{\kappa,\alpha}\right\rangle^2 \right|^{1/2} \left| \bE \left[|Y_3^\kappa|^2 \right] \right|^{1/2} \\
&\leq C(\rho)\tilde{C}_1N \left|\frac{3|a|T}{4} |a|T(\delta + \delta^2)  \right|^{1/2} \left| \bE \left[|Y_3^\kappa|^2 \right] \right|^{1/2}
\leq C(\rho)\frac{\tilde{C}_2}{\kappa^4}|a|T\sqrt{\delta+ \delta^2},
\end{align*}
for some constants $\tilde{C}_1, \tilde{C}_2 > 0$. The last inequality follows from Equation (\ref{e.bds.1}). Similarly, we have
\begin{align*}
&\left| \Tr\ \bE \left[
\left( \cdot, \nu_{R[a,T]}^{\kappa,\rho} \right)_\sharp \left( \cdot, \lambda_\delta^\kappa \right)_\sharp
(Y_1^\kappa+ Y_2^\kappa)^2\right] \right|\leq C(\rho)\frac{\tilde{C}_3}{\kappa^4}|a|T\sqrt{\delta+ \delta^2},
\end{align*}
for some positive constant $\tilde{C}_3$.

Equation (\ref{e.a.9}) also applies, by replacing $\left( \cdot, \nu_{R_\delta[a,T]}^{\kappa,\rho} \right)_\sharp^2$ with $\left( \cdot, \nu_{R[a,T]}^{\kappa,\rho} \right)_\sharp \left( \cdot, \lambda_\delta^\kappa \right)_\sharp$, and we have for large enough $\kappa \geq \kappa_0$, \beq \left|\Tr\ \bE \left[\left( \cdot, \nu_{R[a,T]}^{\kappa,\rho} \right)_\sharp \left( \cdot, \lambda_\delta^\kappa \right)_\sharp(\mathcal{Y}^\kappa-1)\right]\right| \leq C(\rho)\frac{\tilde{C}_4}{\kappa^4}|a|T\sqrt{\delta+ \delta^2}, \nonumber \eeq for some constant $\tilde{C}_4 > 0$.

Hence by choosing $\delta$ small enough, we see from the previous lemma that the trace of
\begin{align*}
\bE &\left[
\left( \cdot, \nu_{R[a,T]}^{\kappa,\rho} \right)_\sharp \left( \cdot, \nu_{R_\delta[a,T]}^{\kappa,\rho} \right)_\sharp
(\mathcal{Y}^\kappa-1)\right] \\
&= \bE \left[
\left( \cdot, \nu_{R[a,T]}^{\kappa,\rho} \right)_\sharp^2
(\mathcal{Y}^\kappa-1)\right] + \bE \left[
\left( \cdot, \nu_{R[a,T]}^{\kappa,\rho} \right)_\sharp \left( \cdot, \lambda_\delta^\kappa \right)_\sharp
(\mathcal{Y}^\kappa-1)\right],
\end{align*}
is bounded above and below for all $\kappa \geq \kappa_0$, $\kappa_0$ dependent on $\delta$, i.e. \beq
\tilde{c}_3\frac{C(\rho)}{\kappa^4} \leq \Tr\ \bE \left[
\left( \cdot, \nu_{R[a,T]}^{\kappa,\rho} \right)_\sharp \left( \cdot, \nu_{R_\delta[a,T]}^{\kappa,\rho} \right)_\sharp
(\mathcal{Y}^\kappa-1)\right] \leq \tilde{c}_4\frac{C(\rho)}{\kappa^4}, \nonumber \eeq for positive constants $\tilde{c}_3, \tilde{c}_4$ independent of $\kappa$ and $\rho$. Together with Equation (\ref{e.a.6}), we will have our result.
\end{proof}

\begin{thm}\label{t.a.1}
Let $c = 1/\kappa$. Then \beq -\Tr\ \bE \left[
\left( \cdot, \nu_{R[a,T]}^{1/c,\rho} \right)_\sharp \left( \cdot, \nu_{R_\delta[a,T]}^{1/c,\rho} \right)_\sharp
\mathcal{Y}^{1/c}\right] = \frac{|a|T}{4} \otimes \Tr\ \mathscr{E}(\rho) - \Lambda NC(\rho) c^4 + C(\rho) f(c), \label{e.s.3} \eeq is continuously differentiable in $c>0$. Furthermore, $\Lambda > 0$ and the error term $f(c) =  O(c^6)$. Hence, there exists a constant $\hat{c} > 0$ such that $|f(c)| \leq \hat{c}c^6$ and $|f'(c)| \leq \hat{c}c^5$.
\end{thm}

\begin{proof}
For each $i=1,2$, $(\kappa^2 Y_i^\kappa, \tilde{\mu}_{\kappa^2}^{\times^N})$, $(\kappa^4 Y_3^\kappa, \tilde{\mu}_{\kappa^2}^{\times^N})$ and $\left(\left( \cdot, \nu_{R_\delta[a,T]}^{1/c,\rho} \right)_\sharp, \tilde{\mu}_{\kappa^2}^{\times^N}\right)$ are equal in distributions to $( Y_i^1, \tilde{\mu}_{1}^{\times^N})$, $( Y_3^1, \tilde{\mu}_{1}^{\times^N})$ and $\left(\left( \cdot, c\nu_{R_\delta[a,T]}^{1/c,\rho} \right)_\sharp, \tilde{\mu}_{1}^{\times^N}\right)$ respectively. Write \beq \mathcal{Z}^c = \exp\left[-\frac{1}{2}\left(c^2Y_1^1 + c^2Y_2^1 + c^4Y_3^1 \right)\right], \nonumber \eeq which is integrable for some $p > 1$ with respect to Wiener measure $\tilde{\mu}_1^{\times^N}$.
From Equations (\ref{e.a.9}) and (\ref{e.a.6}), we see that
\begin{align}
-\Tr\ \bE \left[
\left( \cdot, \nu_{R[a,T]}^{1/c,\rho} \right)_\sharp \left( \cdot, \nu_{R_\delta[a,T]}^{1/c,\rho} \right)_\sharp
\mathcal{Y}^{1/c}\right]
=& -\Tr\ \bE \left[
\left( \cdot, c\nu_{R[a,T]}^{1/c,\rho} \right)_\sharp \left( \cdot, c\nu_{R_\delta[a,T]}^{1/c,\rho} \right)_\sharp
\mathcal{Z}^c\right]  \label{e.c.4}\\
=& \frac{|a|T}{4} \otimes \Tr\ \mathscr{E}(\rho) - \Lambda NC(\rho) c^4 + C(\rho)f(c),\nonumber
\end{align}
for some positive number $\Lambda > 0$ and $f(c) =  O(c^6)$. In the LHS and RHS of Equation (\ref{e.c.4}), the expectation is on the probability space $(\dB \otimes \mathfrak{g}, \tilde{\mu}_{1/c^2}^{\times^{N}})$ and $(\dB \otimes \mathfrak{g}, \tilde{\mu}_{1}^{\times^{N}})$ respectively.
Note that
\begin{align*}
\Lambda =& \frac{1}{2}\lim_{\kappa \rightarrow \infty}\bE\left[\left( \cdot, \frac{1}{\kappa}\nu_{R[a,T]}^{\kappa,1} \right)_\sharp \left( \cdot, \frac{1}{\kappa}\nu_{R_\delta[a,T]}^{\kappa,1} \right)_\sharp \left\{ Y_3^1 - \frac{1}{4}\left(\sum_{i=1}^2 Y_i^1 \right)^2 \right\}\right],
\end{align*}
the expectation taken on the probability space $(\dB \otimes \mathfrak{g}, \tilde{\mu}_{1}^{\times^{N}})$. Indeed, by Corollary \ref{c.c.1}, we see that the correlation between $\left( \cdot, \frac{1}{\kappa}\nu_{R_\delta[a,T]}^{\kappa,\alpha} \right)_\sharp$ and $Y_i$ goes to zero as $\kappa \rightarrow \infty$, and thus \beq
\Lambda = \frac{|a|T}{8}\bE\left[ Y_3^1 - \frac{1}{4}\left(\sum_{i=1}^2 Y_i^1 \right)^2 \right]. \nonumber \eeq

The remainder term $f(c)$ is of the order $c^6$ and for large $\kappa > \kappa_0$, we have the norm of the remainder term is bounded above by $\hat{c} C(\rho)c^6$ for some $\hat{c}>0$. It remains to show that the path integral in Equation (\ref{e.c.4}) is differentiable in $c > 0$.

From Corollary \ref{c.c.2}, define $\left(\cdot, \beta_{R_\delta[a,T]}^{1/c,\rho}\right)_\sharp := \sum_{\alpha=1}^N\left(\cdot, \beta_{R_\delta[a,T]}^{1/c,\alpha}\right)_\sharp \otimes \rho(E^\alpha)$, and we see that
\begin{align*}
\Tr\ &\bE \frac{d}{dc}\left[
\left( \cdot, c\nu_{R[a,T]}^{1/c,\rho} \right)_\sharp \left( \cdot, c\nu_{R_\delta[a,T]}^{1/c,\rho} \right)_\sharp
\mathcal{Z}^c\right] \\
=& \Tr\ \bE \left[
\left( \cdot, \beta_{R[a,T]}^{1/c,\rho} \right)_\sharp \left( \cdot, c\nu_{R_\delta[a,T]}^{1/c,\rho} \right)_\sharp
\mathcal{Z}^c\right] +\Tr\ \bE \left[
\left( \cdot, c\nu_{R[a,T]}^{1/c,\rho} \right)_\sharp \left( \cdot, \beta_{R_\delta[a,T]}^{1/c,\rho} \right)_\sharp
\mathcal{Z}^c\right] \\
&- \Tr\ \bE \left[
\left( \cdot, c\nu_{R[a,T]}^{1/c,\rho} \right)_\sharp \left( \cdot, c\nu_{R_\delta[a,T]}^{1/c,\rho} \right)_\sharp
\mathcal{Z}^c\frac{1}{2}\left(2cY_1^1 + 2cY_2^1 + 4c^3Y_3^1 \right)\right].
\end{align*}
Corollary \ref{c.p.1} says that $\mathcal{Z}^c$ is $L^p$ integrable for some $p > 1$, hence we see that both
\beq \Tr\ \bE \left[
\left|\left( \cdot, \beta_{R[a,T]}^{1/c,\rho} \right)_\sharp \left( \cdot, c\nu_{R_\delta[a,T]}^{1/c,\rho} \right)_\sharp
\mathcal{Z}^c\right|\right] + \Tr\ \bE \left[
\left|\left( \cdot, c\nu_{R[a,T]}^{1/c,\rho} \right)_\sharp \left( \cdot, \beta_{R_\delta[a,T]}^{1/c,\rho} \right)_\sharp
\mathcal{Z}^c\right|\right] < \infty  \nonumber \eeq and
\beq \Tr\ \bE \left[\left|
\left( \cdot, c\nu_{R[a,T]}^{1/c,\rho} \right)_\sharp \left( \cdot, c\nu_{R_\delta[a,T]}^{1/c,\rho} \right)_\sharp \right|
\mathcal{Z}^c\left|Y_1^1 + Y_2^1 + 2c^2Y_3^1 \right|\right] < \infty, \nonumber \eeq using Equation (\ref{e.bds.1}).

The proof of Lemma \ref{l.y.1} shows that \beq \mathcal{Z}^c(B) \leq \exp\left[\sum_{1 \leq i < j \leq 1} \int_{w \in \bC^4}d\lambda_4(w)\left|\left(B, \xi_{ij}(w) \right)_\sharp \right|^2\right]=: \mathcal{T}(B), \nonumber \eeq and $\mathcal{T}$ is $L^p$ integrable for some $p > 1$ with respect to Wiener measure $\tilde{\mu}_1^{\times^N}$.

By Corollary \ref{c.s.1}, we can find a constant $\tilde{C}(c_0)>0$ such that for any $c > c_0 >0$ and $B = \sum_{\alpha=1}^N B_\alpha \otimes E^\alpha \in \dB \otimes \mathfrak{g}$,
\begin{align*}
\tilde{C}(c_0)C(\rho) \sum_{\alpha=1}^N |B_\alpha|^2\geq& \left|\Tr\ \left( B, \beta_{R[a,T]}^{1/c,\rho} \right)_\sharp \left( B, c\nu_{R_\delta[a,T]}^{1/c,\rho} \right)_\sharp \right|  \\
&+ \left|\Tr\ \left( B, c\nu_{R[a,T]}^{1/c,\rho} \right)_\sharp \left( B, \beta_{R_\delta[a,T]}^{1/c,\rho} \right)_\sharp\right| \\
&+ \left|\Tr\ \left( B, c\nu_{R[a,T]}^{1/c,\rho} \right)_\sharp \left( B, c\nu_{R_\delta[a,T]}^{1/c,\rho} \right)_\sharp \right|,
\end{align*}
$|B_\alpha|$ defined in Equation (\ref{e.s.2}). Lemma \ref{l.s.1} says that $\left[\sum_{\alpha=1}^N |B_\alpha|^2 \right]^q$ is integrable with respect to Wiener measure $\tilde{\mu}_{1}^{\times^{N}}$ for any $q \geq 1$. Therefore, for $c_0 < c \leq 1$,
\begin{align*}
& \left|
\Tr\ \left( B, c\nu_{R[a,T]}^{1/c,\rho} \right)_\sharp \left( B, c\nu_{R_\delta[a,T]}^{1/c,\rho} \right)_\sharp
\mathcal{Z}^c\right|+ \left|\frac{d}{dc}\left[\Tr\
\left( B, c\nu_{R[a,T]}^{1/c,\rho} \right)_\sharp \left( B, c\nu_{R_\delta[a,T]}^{1/c,\rho} \right)_\sharp
\mathcal{Z}^c\right] \right| \\
&\leq \mathcal{T}(B)\tilde{C}(c_0)C(\rho)\left[1 + \left|Y_1^1(B) + Y_2^1(B) + 2Y_3^1(B)\right| \right] \sum_{\alpha=1}^N |B_\alpha|^2,
\end{align*}
the RHS is integrable with respect to Wiener measure $\tilde{\mu}_{1}^{\times^{N}}$.

Hence, we can interchange $d/dc$ with $\bE$, expectation taken on a fixed probability space $(\dB \otimes \mathfrak{g}, \tilde{\mu}_{1}^{\times^{N}})$, and thus
\begin{align*}
\frac{d}{dc}&\Tr\ \bE \left[
\left( \cdot, c\nu_{R[a,T]}^{1/c,\rho} \right)_\sharp \left( \cdot, c\nu_{R_\delta[a,T]}^{1/c,\rho} \right)_\sharp
\mathcal{Z}^c\right] = \Tr\ \bE \frac{d}{dc}\left[
\left( \cdot, c\nu_{R[a,T]}^{1/c,\rho} \right)_\sharp \left( \cdot, c\nu_{R_\delta[a,T]}^{1/c,\rho} \right)_\sharp \mathcal{Z}^c\right]\\
=& \Tr\ \bE\left[
\left( \cdot, \beta_{R[a,T]}^{1/c,\rho} \right)_\sharp \left( \cdot, c\nu_{R_\delta[a,T]}^{1/c,\rho} \right)_\sharp
\mathcal{Z}^c\right] +\Tr\ \bE \left[
\left( \cdot, c\nu_{R[a,T]}^{1/c,\rho} \right)_\sharp \left( \cdot, \beta_{R_\delta[a,T]}^{1/c,\rho} \right)_\sharp
\mathcal{Z}^c\right] \\
&- c\Tr\ \bE \left[
\left( \cdot, c\nu_{R[a,T]}^{1/c,\rho} \right)_\sharp \left( \cdot, c\nu_{R_\delta[a,T]}^{1/c,\rho} \right)_\sharp
\mathcal{Z}^c\left(Y_1^1 + Y_2^1 + 2c^2Y_3^1 \right)\right].
\end{align*}
Since $c_0 > 0$ is arbitrary, this shows that the error term $f(c)$ is continuously differentiable for $c > 0$. From $f(c) = O(c^6)$, we see that $g(c) = \frac{f(c)}{c^5}$ is Lipschitz continuous in $c \geq 0$, hence $\limsup_{c \rightarrow 0^+}|g'(c)|$ must be bounded and thus we have a constant $\hat{c}$ such that $|f'(c)| \leq \hat{c} c^5$. This completes the proof.
\end{proof}

\begin{rem}\label{r.cz.1}
We do not need the differentiability property of the path integral in Equation (\ref{e.s.3}). We will however, need this fact in a sequel to this article, when we need to compute the beta function from the Callan-Symanzik Equation.
\end{rem}

\begin{defn}
Let $N_n:= \{1,\cdots, 2n\}$, and we will pair the numbers in $N_n$, represented as $f_n = \{(\underline{\tau}_1, \overline{\tau}_1), (\underline{\tau}_2, \overline{\tau}_2), \cdots, (\underline{\tau}_n, \overline{\tau}_n)\}$, whereby
\begin{itemize}
  \item for $1 \leq k \leq n$, $\underline{\tau}_k$ and $\overline{\tau}_k$ take values in $N_n$ and they are all distinct from each other;
  \item $\underline{\tau}_k < \overline{\tau}_k$ for every $k$; and
  \item $1 = \underline{\tau}_1 < \underline{\tau}_2 < \cdots < \underline{\tau}_n$.
\end{itemize}
We say $f_n$ is an ordered pairing of $\{1, \cdots, 2n\}$.  Let $S_n$ be the set containing all possible ordered pairings of $\{1, \cdots, 2n\}$ satisfying the above conditions.

Given a representation $\rho: \mathfrak{g} \rightarrow {\rm End}(\bC^{\tilde{N}})$, and a set of matrices $\{\rho(E^\alpha): 1 \leq \alpha \leq N\}$, let $\tilde{k}_n = (a_1, a_2, \cdots, a_n)$, whereby each $a_i$'s take natural numbers from 1 to $N$. Note that $a_i$'s need not all be distinct. Let \beq K_n := \left\{ \tilde{k}_n = (a_1, a_2, \cdots, a_n):\ 1 \leq a_i \leq N \right\}. \nonumber \eeq

Given a $\tilde{k}_n$ and $f_n$ as denoted above, we will write \beq G_\rho(\tilde{k}_n, f_n) := \rho(E^{\alpha_1})\rho(E^{\alpha_2})\cdots \rho(E^{\alpha_{2n}}), \nonumber \eeq whereby $\rho(E^{\alpha_s}) = \rho(E^\beta)$, $1 \leq s \leq 2n$, if
\begin{itemize}
  \item $s$ is equal to either $\underline{\tau}_r$ or $\overline{\tau}_r$ for a unique $1 \leq r \leq n$ in $f_n$;
  \item $\beta = a_r $ in $\tilde{k}_n$.
\end{itemize}
When $n = 0$, we will write $G_\rho(\tilde{k}_0, f_0) = I_{\tilde{N}}$, the $\tilde{N} \times \tilde{N}$ identity matrix, with $\tilde{N}$ as the dimension of the representation.
\end{defn}

\begin{thm}\label{t.p.2}
Refer to Definitions \ref{d.ym.2}, \ref{d.r.1} and \ref{d.w.4}. As $\kappa \rightarrow \infty$, we have
\begin{align*}
\mathbb{E}_{{\rm YM}}^\kappa&\left[ \mathcal{J}_{R[a,T]}^{\kappa,\rho} \right] \\
&\longrightarrow
\sum_{n=0}^\infty \frac{1}{(2n)!}\left[ \frac{1}{4}\sum_{0\leq i\leq 3}\int_{I^2} ds dt\ \rho_{\sigma}^{0i}(s,t)J_{0i}^\sigma(s,t)\right]^n \sum_{\tilde{k}_n \in K_n}\sum_{f_n \in S_n}G_\rho(\tilde{k}_n, f_n) .
\end{align*}
Here, \beq \sum_{0\leq i\leq 3}\int_{I^2} ds dt\ \rho_{\sigma}^{ab}(s,t)J_{0i}^\sigma(s,t) = |a|T, \nonumber \eeq which is the area of the rectangular surface $R[a,T]$.
\end{thm}

\begin{proof}

Recall $\mathfrak{g}$ is the semi-simple Lie Algebra of a compact Lie group $G$ which is a subgroup of $U(\bar{N})$. We also assume that for an irreducible representation $\rho: \mathfrak{g} \rightarrow {\rm End}(\bC^{\tilde{N}})$, $\rho(E^\alpha)$ is skew Hermitian. Hence, for any $B \in \dB \otimes \mathfrak{g}$, $\mathcal{J}_{R[a,T]}^{\kappa,\rho}(B) = \exp\left[\left( B, \nu_{R[a,T]}^{\kappa,\rho} \right)_\sharp \right]$ is unitary, thus its spectral norm $\left|\mathcal{J}_{R[a,T]}^{\kappa,\rho}(B)\right|$ is bounded by some constant $C$. We also have its spectral norm $|\rho(E^\alpha)\rho(E^\alpha)| \leq C(\rho)$, or $|\rho(E^\alpha)| \leq \sqrt{C(\rho)}$.

The proof in Lemma \ref{l.bds.1} can be adapted to show that \beq \Tr\ \mathbb{E}\left[\left|\mathcal{J}_{R[a,T]}^{\kappa,\rho} (1 - \mathcal{Y}^\kappa)\right| \right] = O(1/\kappa^2), \nonumber \eeq by replacing $\left( \cdot, \nu_{R[a,T]}^{\kappa,\rho} \right)_\sharp^2$ with $\left|\mathcal{J}_{R[a,T]}^{\kappa,\rho}\right|$.

Hence we have
\begin{align*}
\lim_{\kappa \rightarrow \infty}\int_{\dB \otimes \mathfrak{g}}\mathcal{J}_{R[a,T]}^{\kappa,\rho} \mathcal{Y}^\kappa\ d\tilde{\mu}_{\kappa^2}^{\times^{N}} &= \lim_{\kappa \rightarrow \infty}\mathbb{E}\left[\mathcal{J}_{R[a,T]}^{\kappa,\rho} \right].
\end{align*}

For a normal random variable $N$ with variance $\alpha^2$, an induction argument will show that $\bE[ N^{2n}] = \alpha^{2n}(2n-1)(2n-3) \cdots 3 \times 1$. And
\begin{align*}
\bE[|N|^{2n+1}] &= \bE[|N|^{2n+1}1_{\{|N| < 1\}}] + \bE[|N|^{2n+1}1_{\{|N| \geq 1\}}] \\
&\leq 1 + \bE[|N|^{2(n+1)}] \\
&= 1 + \alpha^{2(n+1)}(2n+1)(2n-1)\cdots 3.
\end{align*}
Hence, we see that
\begin{align}
\bE&\left[\sum_{n=0}^\infty \frac{1}{n!} |N|^n \left(\sqrt{C(\rho)} \right)^n \right] =
\sum_{n=0}^\infty \left(\sqrt{C(\rho)} \right)^n\frac{1}{n!} \bE[|N|^n] \nonumber \\
\leq& 1+ \sum_{n=1}^\infty \left(\sqrt{C(\rho)} \right)^{2n}\alpha^{2n} \frac{1}{(2n)!}[(2n-1)(2n-3) \cdots 3 \times 1] \nonumber\\
&+ \sum_{n=0}^\infty \left(\sqrt{C(\rho)} \right)^{2n+1}\frac{1}{(2n+1)!}\left[1 + \alpha^{2(n+1)}(2n+1)(2n-1)\cdots 3 \right] \nonumber \\
\leq& 1+\sum_{n=1}^\infty C(\rho)^{n}\alpha^{2n} \frac{1}{n!}\frac{1}{2^n} + \sum_{n=0}^\infty \left(\sqrt{C(\rho)} \right)^{2n+1}[1+\alpha^{2(n+1)}]\frac{1}{n!}\frac{1}{2^n} < \infty. \label{e.b.6}
\end{align}

Refer to Definition \ref{d.p.2}. Now, $\left\{\left( \cdot ,\nu_{R[a,T]}^{\kappa, \alpha}\right)_\sharp: 1 \leq \alpha \leq N \right\}$ is a set of i.i.d. real Gaussian random variables. From Equation (\ref{e.a.8}), we can apply Dominated Convergence Theorem, and we have
\begin{align*}
\bE&\left[ e^{\left(\cdot, \nu_{R[a,T]}^{\kappa, \rho} \right)_\sharp} \right]= \bE\sum_{n=0}^\infty \frac{1}{n!}\left[\sum_{\alpha=1}^N\left(\cdot, \nu_{R[a,T]}^{\kappa, \alpha} \right)_\sharp\otimes \rho(E^\alpha)\right]^{n}\\
&= \sum_{n=0}^\infty \frac{1}{n!}\bE\left[\sum_{\alpha=1}^N\left(\cdot, \nu_{R[a,T]}^{\kappa, \alpha} \right)_\sharp\otimes \rho(E^\alpha)\right]^{n}\\
&= \sum_{n=0}^\infty \frac{1}{(2n)!} \sum_{\tilde{k}_n=(\alpha_1,\cdots, \alpha_n) \in K_n}\sum_{f_n \in S_n}\prod_{i=1}^n\frac{1}{\kappa^2}\left\langle \nu_{R[a,T]}^{\kappa, \alpha_i}, \nu_{R[a,T]}^{\kappa, \alpha_i}\right\rangle G_\rho(\tilde{k}_n, f_n) \\
&\longrightarrow \sum_{n=0}^\infty \frac{1}{(2n)!} \sum_{\tilde{k}_n \in K_n}\sum_{f_n \in S_n}\left[ \frac{1}{4}\sum_{1\leq i \leq 3}\int_{I^2} ds dt\ \rho_\sigma^{0i}(s,t)|J_{0i}^\sigma|(s,t)  \right]^n G_\rho(\tilde{k}_n, f_n),
\end{align*}
as $\kappa \rightarrow \infty$. Because of Equation (\ref{e.b.6}), we can apply Fubini's Theorem to interchange $\bE$ with the infinite sum $\sum_{n=0}^\infty$.

From the proof of Lemma \ref{l.bds.1}, we see that
\begin{align*}
\bE[\mathcal{Y}^\kappa] =  1+ O(1/\kappa^4),
\end{align*}
which converges to 1 as $\kappa$ goes to infinity.

Hence,
\begin{align*}
&\mathbb{E}_{{\rm YM}}^\kappa\left[ \mathcal{J}_{R[a,T]}^{\kappa,\rho} \right] = \frac{1}{\bE[\mathcal{Y}^\kappa]}\bE\left[\mathcal{J}_{R[a,T]}^{\kappa,\rho} \mathcal{Y}^\kappa \right] \\
&\longrightarrow \sum_{n=0}^\infty \frac{1}{(2n)!} \sum_{\tilde{k}_n \in K_n}\sum_{f_n \in S_n}\left[ \frac{1}{4}\sum_{1\leq i \leq 3}\int_{I^2} ds dt\ \rho_\sigma^{0i}(s,t)|J_{0i}^\sigma|(s,t)  \right]^n G_\rho(\tilde{k}_n, f_n),
\end{align*}
as $\kappa \rightarrow \infty$.
\end{proof}

\begin{rem}\label{r.af.1}
If one examines the above proof carefully, asymptotic freedom is required to show that $\lim_{\kappa \rightarrow \infty}\mathbb{E}\left(\exp\left[\left(\cdot, \nu_{R[a,T]}^{\kappa, \alpha}\right)_\sharp \right]\right)$ exists and is equal to the exponential of the area of the surface. Asymptotic freedom and renormalization allow one to compute the path integral in Definition \ref{d.ym.2} via perturbation series expansion, when $\kappa$ is large.

In Remark \ref{r.aym.1}, we explained that $1/\kappa$ denote the spread of a Gaussian function, used to approximate a Dirac Delta function, whereby $1/\kappa$ is also synonymous with the lattice spacing $\epsilon$ in lattice gauge theories. When the distance between distinct points is small, one needs a higher $\kappa$, to resolve them. Thus, we can say that non-abelian Yang-Mills theory is asymptotically free at short distances. In a sequel, we will use Theorem \ref{t.a.1} to show that the theory is asymptotically free at high energies. See Remark \ref{r.cz.1}.

We need to normalize $\chi_w$, and choose a $0<\tilde{c}<1$ such that the norm $\parallel \tilde{c}\psi_w \chi_w \parallel < 1$. In this article, we chose $\tilde{c} = 1/\sqrt{2\pi}$. In \cite{YMLim01}, we explained why we need the renormalization factor $\frac{1}{\sqrt{2\pi}}\psi_w$. In Quantum Field Theory, when using Fourier Transform to go into momentum space, one will introduce a cutoff function, so as to resolve the ultra-violet divergence issue. This was sufficient to prove Corollary \ref{c.p.1} and Lemma \ref{l.bds.1}. Also refer to Remarks \ref{r.p.1} and \ref{r.bds.1}. Indeed, Corollary \ref{c.p.1} and Lemma \ref{l.bds.1} still hold by choosing any $\tilde{c} < 1/\sqrt2$.
\end{rem}

\section{Application}\label{s.app}

The main purpose of computing the Yang-Mills path integral is to derive the Wilson Area Law formula. A quark is never observed in isolation due to quark confinement. A quark and an antiquark are bounded together in a meson via the strong force. Unlike the Coulomb potential, which decays inversely proportional to the distance, the strong force actually gets stronger as the separation between the quarks increases. It is confirmed experimentally that the potential varies proportionally to the distance between quarks.

The path integral given by Expression \ref{e.ym.1} can be used to explain why this potential is linear. However, to compute it is highly nontrivial and difficult. Note that perturbative methods using Feynman diagrams cannot be used to compute it, lest derive the area formula.

Let  $S$ be a flat rectangular surface, whose boundary is a rectangular contour $C$ with spatial length $R$ and timelike length $T$, thus $S$ has area $RT$. Then, we will write \beq
W(R, T; \rho) := \lim_{\kappa \rightarrow \infty}\frac{1}{\tilde{N}}\bE_{{\rm YM}}^\kappa\left[\Tr\ \exp\left[(\cdot, \nu_{S}^{\kappa,\rho})_\sharp \right]\right]. \nonumber \eeq From Theorem \ref{t.p.2}, we have the average trace \beq W(R, T; \rho) = \frac{1}{\tilde{N}}\Tr\ \sum_{n=0}^\infty \frac{1}{(2n)!} \sum_{\tilde{k}_n\in K_n}\sum_{f_n \in S_n}\left[\frac{RT}{4}\right]^n G_\rho(\tilde{k}_n, f_n). \nonumber \eeq

A meson is made up of a quark and antiquark. Now, these 2 quarks go in opposite direction along the spatial length $R$ and after time $T$, are attracted to each other along the spatial length $R$, hence tracing a rectangular contour $C$ described in the previous paragraph. We want to calculate the potential energy between the quark and antiquark, as a function of distance $R$. Using the path integral computed on the surface $S$ with boundary $C$, the potential is given by \beq V(R) = -\lim_{T \rightarrow 0}\frac{1}{T}\log W(R, T; \rho). \nonumber \eeq

When $n = 1$, note that $\sum_{\tilde{k}_1\in K_1}\sum_{f_1 \in S_1}G_\rho(\tilde{k}_1, f_1) = -\mathscr{E}(\rho)$, the negative of the Casimir operator. When $RT$ is small, we have that $W(R, T; \rho) = 1 - \frac{RT}{8}\frac{1}{\tilde{N}}\Tr\ \mathscr{E}(\rho) + O(R^2T^2)$ and thus \beq V(R) = \frac{R}{8\tilde{N}}\Tr\ \mathscr{E}(\rho). \nonumber \eeq

In the case of the standard representation of ${\rm SU}(N)$, we can see that \beq  \sum_{\alpha=1}^{N^2-1} E^\alpha  E^\alpha  = \left(\frac{1}{N} - N\right)\mathbb{I}, \nonumber \eeq $\mathbb{I}$ is the $N \times N$ identity matrix. See \cite{CSLim02}. Therefore, \beq V(R) = \frac{R}{8}\left( N - \frac{1}{N} \right). \nonumber \eeq

In the case of the standard representation of ${\rm SO}(N)$, we see that \beq \sum_{\alpha=1}^{N(N-1)/2} E^\alpha  E^\alpha  = \frac{1-N}{2}\mathbb{I} , \nonumber \eeq hence
\beq V(R) = \frac{R}{16}\left( N - 1 \right). \nonumber \eeq Compare with Equation (19.6) in \cite{Nair}. Thus, we see that quark confinement implies that the potential energy between the quarks is linear in nature.

\appendix

\section{Definition}

Let $S$ be a surface embedded in $\bR^4$ and $\sigma\equiv ( \sigma_0, \sigma_1, \sigma_2, \sigma_3): [0,1]^2 \equiv I^2 \rightarrow \bR^4$ be its parametrization. Here, $\sigma' = \partial \sigma/\partial s$ and $\dot{\sigma} = \partial \sigma/\partial t$.

\begin{defn}\label{d.r.1}
For $a,b=0,1,2,3$, define Jacobian matrices,
\begin{align}
J_{ab}^\sigma(s,t)
= \left(
\begin{array}{cc}
\sigma_a'(s,t) &\ \dot{\sigma}_a(s,t) \\
\sigma_b'(s,t) &\ \dot{\sigma}_b(s,t) \\
\end{array}
\right),\ a \neq b, \nonumber
\end{align}
and write $|J^\sigma_{ab}| = \sqrt{[\det{J^\sigma_{ab}}]^2}$ and $W_{ab}^{ cd} := J_{cd}^\sigma J_{ab}^{\sigma, -1}$, $a, b, c, d$ all distinct. Note that $W_{cd}^{ab} = (W_{ab}^{cd})^{-1}$.

For $a,b, c, d$ all distinct, define $\rho_\sigma^{ab}: I^2 \rightarrow \bR$ by
\begin{align}
\rho_\sigma^{ab} =& \frac{1}{\sqrt{\det\left[ 1+ W_{ab}^{cd,T}W_{ab}^{cd}\right]}} \equiv \frac{|J_{ab}^\sigma|}{\sqrt{\det\left[J_{ab}^{\sigma,T}J_{ab}^\sigma + J_{cd}^{\sigma,T}J_{cd}^\sigma\right]}}. \label{e.ri.1}
\end{align}
\end{defn}

\begin{defn}\label{d.w.4}
Define
\begin{align}
\int_S d\rho := \sum_{0 \leq a<b \leq 3}\int_{I^2}\rho_\sigma^{ab}(s,t)|J_{ab}^\sigma|(s,t)
\ ds dt, \nonumber
\end{align}
which gives us the area of the surface $S$.
\end{defn}

\begin{rem}
\begin{enumerate}
  \item Note that $\int_S d\rho$ is independent of the choice of orthonormal basis and parametrization used, by a straightforward calculation. See \cite{MR1312606}.
  \item Dimension 4 is required to write the formula for $\int_S d\rho$.
\end{enumerate}
\end{rem}

\section{Correlation}

Recall we defined $\zeta(w)$ in Proposition \ref{p.z.1}. By its definition, we can write for $w = (w_0, w_1, w_2, w_3)\in \bC^4$, $\zeta(w)(z) = g_{w_0}(z_0)e^{\sum_{i=1}^3 \bar{w}_i z_i}$, $\bar{w}_i$ is its complex conjugate. Note that $g_{w_0} \in \mathcal{H}(\bC^4)_\bC$.

\begin{prop}\label{p.s.3}
Let $\sigma_0: t \in I_\delta \equiv [-\delta, 1+\delta]  \longmapsto tT$, for some $T \neq 0$. We have
\beq
\lim_{\kappa \rightarrow \infty}g_{w_0}(\kappa \sigma_0(t)/2)e^{-\kappa^2|\sigma_0(t)/2|^2/2}=0, \nonumber \eeq for all $t \neq 0$.
\end{prop}

\begin{proof}
Write $g_{w_0}(z_0) = \sum_{n=0}^\infty c_n(w_0)\frac{z_0^n}{\sqrt{n!}}$. Because $g_{w_0} \in \mathcal{H}^2(\bC^4)\otimes \bC$, we see that $\sum_{n=0}^\infty |c_n(w_0)|^2 < \infty$. Let $\epsilon > 0$. Choose a $N(\epsilon)$ such that $\sum_{n \geq N(\epsilon)}|c_n(w_0)|^2 < \epsilon^2$. By Cauchy Schwartz Inequality,
\begin{align*}
&|g_{w_0}(\kappa \sigma_0(t)/2)| \leq \left|\sum_{n=0}^{N(\epsilon)-1}c_n(w_0)\frac{(\kappa \sigma_0(t)/2)^n}{\sqrt{n!}} \right| + \left|\sum_{n\geq N(\epsilon)} c_n(w_0)\frac{(\kappa \sigma_0(t)/2)^n}{\sqrt{n!}}\right| \\
&\leq \left|\sum_{n=0}^{N(\epsilon)-1}c_n(w_0)\frac{(\kappa \sigma_0(t)/2)^n}{\sqrt{n!}} \right| + \left(\sum_{n \geq N(\epsilon)}|c_n(w_0)|^2 \right)^{1/2}\left( \sum_{n \geq N(\epsilon)}\frac{(\kappa \sigma_0(t)/2)^{2n}}{n!}\right)^{1/2} \\
&\leq \left|\sum_{n=0}^{N(\epsilon)-1}c_n(w_0)\frac{(\kappa \sigma_0(t)/2)^n}{\sqrt{n!}} \right| + \epsilon e^{|\kappa \sigma_0(t)/2|^2/2}.
\end{align*}
Hence, \beq |g_{w_0}(\kappa \sigma_0(t)/2)|e^{-\kappa^2|\sigma_0(t)/2|^2/2} \leq \left|\sum_{n=0}^{N(\epsilon)-1}c_n(w_0)\frac{(\kappa \sigma_0(t)/2)^n}{\sqrt{n!}} \right|e^{-|\kappa \sigma_0(t)/2|^2/2} + \epsilon , \nonumber \eeq
thus for $t \neq 0$, \beq \lim_{\kappa \rightarrow \infty}|g_{w_0}(\kappa \sigma_0(t)/2)|e^{-\kappa^2|\sigma_0(t)/2|^2/2} \leq \epsilon. \nonumber \eeq Since $\epsilon > 0$ is arbitrary, therefore, $\lim_{\kappa \rightarrow \infty}g_{w_0}(\kappa \sigma_0(t)/2)e^{-\kappa^2|\sigma_0(t)/2|^2/2}=0$ for all $t \neq 0$.
\end{proof}

\begin{lem}
Let $\mathbf{x} = (x^0, x) \in \bR^4$, $x \in \bR^3$. Then, we have for any \\
$\mathbf{w} = (w_0, w_1, w_2, w_3) \in \bC^4$,
\begin{align*}
\frac{1}{2\pi}&e^{-|\mathbf{w}|^2/2}e^{-\kappa^2|\mathbf{x}|^2/8}\langle \zeta(\mathbf{w}), \chi_{\kappa \mathbf{x}/2} \rangle \\
&=\frac{1}{2\pi}e^{-\sum_{j=1}^3|a_j - \kappa x_j/2|^2/2}e^{-(|w_0|^2 + \sum_{j=1}^3 b_j^2)/2}g_{w_0}(\kappa x^0/2)e^{-\kappa^2|x^0|^2/8}e^{-i\kappa \sum_{j=1}^3 b_jx_j/2},
\end{align*}
whereby $w_i = a_i + \sqrt{-1}b_i \in \bC$.
\end{lem}

\begin{proof}
Write $w = (w_1, w_2, w_3) \in \bC^3$, and $w_i = a_i + \sqrt{-1}b_i$. By direct computation, we have \beq \langle \zeta(\mathbf{w}), \chi_{\kappa \mathbf{x}/2} \rangle = g_{w_0}(\kappa x^0/2)\langle e^{\bar{w}\cdot z}, e^{\kappa x \cdot z/2} \rangle = g_{w_0}(\kappa x^0/2) e^{\kappa \bar{w}\cdot  x/2}. \nonumber \eeq Write $f = e^{-(|w_0|^2 + \sum_{j=1}^3 b_j^2)/2}$. Thus,
\begin{align*}
\frac{1}{2\pi}&e^{-|\mathbf{w}|^2/2}e^{-\kappa^2|\mathbf{x}|^2/8}g_{w_0}(\kappa x^0/2) e^{\kappa \bar{w}\cdot  x/2} \\
=& \frac{1}{2\pi}e^{-\sum_{j=1}^3|a_j|^2/2}e^{-\kappa^2|\mathbf{x}|^2/8}fg_{w_0}(\kappa x^0/2)e^{\kappa\sum_{j=1}^3 a_j x_j/2}e^{-i\kappa \sum_{j=1}^3 b_jx_j/2} \\
=& \frac{1}{2\pi}e^{-\sum_{j=1}^3|a_j - \kappa x_j/2|^2/2}fg_{w_0}(\kappa x^0/2)e^{-\kappa^2|x^0|^2/8}e^{-i\kappa \sum_{j=1}^3 b_jx_j/2}.
\end{align*}
\end{proof}

\begin{cor}\label{c.c.1}
We have that for any $1 \leq \alpha, \beta \leq N$,
\begin{align*}
\lim_{\kappa \rightarrow \infty}&\int_{w \in \bC^4 }\left\langle \pi_{i,\alpha}(w), \frac{1}{\kappa}\nu_{R_\delta[a,T]}^{\kappa,\alpha}, \right\rangle \left\langle \pi_{j,\beta}(w), \frac{1}{\kappa}\nu_{R_0[a,T]}^{\kappa,\beta}\right\rangle\ d\lambda_4(w) = 0.
\end{align*}
\end{cor}

\begin{proof}
Recall that \beq \nu_{R_\delta[a,T]}^{\kappa,\alpha}=\frac{1}{\kappa}\int_{I_\delta^2}dsdt\ \frac{\kappa^2}{4}\sum_{j=1}^3 |J_{0j}^{\sigma}|(s,t)  \frac{\kappa e^{-\frac{\kappa^2\left|\sigma(s,t)\right|^2}{8}}}{\sqrt{2\pi}}\chi_{\kappa \sigma(s,t)/2}\otimes dx^0 \wedge dx^j \otimes E^\alpha, \nonumber \eeq for some parametrization $\sigma$ on $I_\delta^2$. We may assume that $\sigma(s,t) = s\tilde{a} + t e^0$, $-\delta \leq s, t \leq 1 + \delta$, whereby $\tilde{a} \in \bR^3$ and $e^0$ is a vector along the time axis. Thus, $\sigma_0(s,t) = \sigma_0(t) = tT$, $(\sigma_1, \sigma_2, \sigma_3)(s,t) = (\sigma_1, \sigma_2, \sigma_3)(s) = s \tilde{a}$. And $|J_{0j}^{\sigma}|(s,t) \equiv |J_{0j}^{\sigma}|$ is a constant, with $\pi_{i,\alpha}(w) = \frac{e^{-|w|^2/2}}{\sqrt{2\pi}}\zeta(w)\otimes dx^0 \wedge dx^i \otimes E^\alpha$. Write $w = (w_0, w_1, w_2, w_3)$, $w_j = a_j + \sqrt{-1}b_j$, $1 \leq j \leq 3$.

Recall from the previous lemma, $f = e^{-(|w_0|^2 + \sum_{j=1}^3 b_j^2)/2}$. Thus,
\begin{align*}
f^{-1}&\Big\langle  \pi_{i,\alpha}(w),\frac{1}{\kappa}\nu_{R_\delta[a,T]}^{\kappa,\alpha}\Big\rangle  = f^{-1}\int_{I_\delta^2}dsdt\ \frac{\kappa|J_{0i}^{\sigma}|}{4} \left\langle \pi_{i,\alpha}(w), \xi_{0i}(\kappa \sigma(s,t)/2)\otimes E^\alpha \right\rangle  \\
=& \int_{I_\delta}ds\ \frac{\kappa}{4} \frac{|J_{0i}^{\sigma}|}{2\pi}e^{-\frac{1}{2}\sum_{j=1}^3\left|a_j - \frac{\kappa \sigma_j(s)}{2}\right|^2}e^{-i\kappa \sum_{j=1}^3 \frac{b_j\sigma_j(s)}{2}}\int_{I_\delta}dt\ g_{w_0}\left(\frac{\kappa \sigma_0(t)}{2}\right)e^{-\frac{\kappa^2|\sigma_0(t)|^2}{8}} \\
=& \int_{-\delta \kappa}^{\kappa(1+\delta)}ds\  \frac{|J_{0i}^{\sigma}|}{8\pi}e^{-\frac{\sum_{j=1}^3\left|a_j - \frac{ \sigma_j(s)}{2}\right|^2}{2}}e^{-i \sum_{j=1}^3 \frac{b_j\sigma_j(s)}{2}} \int_{I_\delta}dt\ g_{w_0}\left(\frac{\kappa \sigma_0(t)}{2}\right)e^{-\frac{\kappa^2|\sigma_0(t)|^2}{2}} \\
\longrightarrow& \frac{1}{8\pi}\int_{\bR}ds\  |J_{0i}^{\sigma}|e^{-\sum_{j=1}^3|a_j - \sigma_j(s)/2|^2/2}e^{-i \sum_{j=1}^3 b_j\sigma_j(s)/2} \times \int_{I_\delta} 0\ dt = 0,
\end{align*}
as $\kappa \rightarrow \infty$, using Proposition \ref{p.s.3}.

We need to justify \beq \lim_{\kappa \rightarrow \infty}\int_{I_\delta}dt\ g_{w_0}(\kappa \sigma_0(t)/2)e^{-\kappa^2|\sigma_0(t)/2|^2/2} = \int_{I_\delta}dt\ \lim_{\kappa \rightarrow \infty}g_{w_0}(\kappa \sigma_0(t)/2)e^{-\kappa^2|\sigma_0(t)/2|^2/2}, \label{e.b.5} \eeq in the last step. First, \beq
|g_{w_0}(\kappa x_0/2)| = |\langle g_{w_0}, e^{\kappa x_0 z_0/2} \rangle| \leq \parallel g_{w_0} \parallel\cdot \parallel e^{\kappa x_0 z_0/2} \parallel = \parallel g_{w_0} \parallel e^{\kappa^2 |x_0|^2/8}, \nonumber \eeq and thus
$|g_{w_0}(\kappa \sigma_0(t)/2)e^{-\kappa^2|\sigma_0(t)/2|^2/2}|$ is bounded by $\sqrt{40}$ from Remark \ref{r.rem.2}, for any $t \in I_\delta$. Hence, Bounded Convergence Theorem applies to Equation (\ref{e.b.5}).

Since $R_\delta[a,T]$ is bounded, Remark \ref{r.rem.2} will also show that $\left|\left\langle \pi_{i,\alpha}(w), \frac{1}{\kappa}\nu_{R_\delta[a,T]}^{\kappa,\alpha}, \right\rangle \right|$ is bounded by a constant, independent of $w$. Hence, we can apply Dominated Convergence Theorem and bring the limit inside $\int_{w\in \bC^4}d\lambda_4(w)$, to obtain the result.
\end{proof}

\section{Derivative}

\begin{defn}\label{d.s.1}
Let $\mathcal{P}_r $ be the set consisting of 4-tuples of the form $p_r = (i,j,k,l)$, each entry inside $\mathbb{N} \cup \{0\}$, such that the sum of its entries $i+j+k+l = r$. For $z = (z_0, z_1, z_2, z_3) \in \bC^4$, write $z^{p_r} = z_0^i z_1^j z_2^k z_3^l$, $p_r = (i,j,k,l)$. We will also write $p_r! := i!j!k!l!$.

For $f(z) = \sum_{r \geq 0} \sum_{p_r \in \mathcal{P}_r} c_{p_r}\frac{z^{p_r}}{\sqrt{p_r!}} \in B(\bC^4)$, define a supremum norm $|\cdot|$, \beq |f| := \sup_{z \in B(0, 1/2)}\sum_{r \geq 0}\sum_{p_r \in \mathcal{P}_r} |c_{p_r}||z^{p_r}|, \label{e.sup.1} \eeq $B(0,1/2)$ is the open ball with radius $1/2$, center $0$ in $\bC^4$.
\end{defn}

Consider $g(z_0) = \sum_{n=0}^\infty c_n z_0^n \in B(\bC^4)$, $z_0 \in \bC$, which is analytic. We saw that $\chi_{w_0}(z_0) = e^{\bar{w}_0z_0}$ defines a linear functional $\left(\cdot, \chi_{w_0}\right)_\sharp$ that maps $g \mapsto g(w_0) = \left(g, \chi_{w_0}\right)_\sharp$.

For a $x \in \bR$, consider $\chi_{x/2c}$, $c$ is some real positive parameter. Then, $g(x/2c) = (g, \chi_{x/2c})_\sharp = \sum_{n=0} c_n \frac{1}{c^n}\left(\frac{x}{2} \right)^n$ and its derivative is given by \beq \frac{d}{dc}g(x/2c) = -\frac{1}{c}\sum_{n=1}^\infty nc_n\left( \frac{x}{2c} \right)^n. \nonumber \eeq

\begin{prop}\label{p.s.1}
Let \beq \theta_{x/2c}(z_0) = \frac{d}{dc}\chi_{x/2c}(z_0) = -\frac{xz_0}{2c^2}e^{xz_0/2c} = -\frac{1}{c}\sum_{n=1}^\infty n \left(\frac{x}{2c} \right)^n \frac{z_0^n}{n!}. \nonumber \eeq For $y \in \bR$ and $\tilde{g} \in B(\bC^4)$, \beq \frac{d}{dc}\tilde{g}(y/2c) = \left\langle \tilde{g}, \theta_{y/2c} \right\rangle. \nonumber \eeq
Indeed, we have for $\mathbf{x} = (x^0, x^1, x^2, x^3)$, $p_r = (i,j,k,l)$ and $z^{p_r} = z_0^i z_1^j z_2^k z_3^l$,
\begin{align*}
\left(\cdot, \theta_{x^0/2c}\chi_{x^1/2c}\chi_{x^2/2c}\chi_{x^3/2c}\right)_\sharp: g(z) &= \sum_{r \geq 0} \sum_{p_r \in \mathcal{P}_r} c_{p_r}\frac{z^{p_r}}{\sqrt{p_r!}} \in B(\bC^4) \\
&\longmapsto -\frac{1}{c}\sum_{r=0}^\infty\sum_{p_r \in \mathcal{P}_r} c_{p_r}\frac{i}{(2c)^r}\frac{\mathbf{x}^{p_r}}{\sqrt{p_r!}}\in B(\bC^4),
\end{align*}
is a bounded linear functional.
\end{prop}

\begin{proof}
Suppose $g(z_0) =\sum_{n=0}^\infty c_n z_0^n\in \mathcal{H}^2(\bC^4)$, a dense subspace in $B(\bC^4)$. Using the fact that $\langle z_0^n, z_0^m\rangle = 0$ if $m \neq n$, and is equal to $n!$ if $n = m$, we have that
\begin{align*}
\frac{d}{dc}g(y/2c) &= -\frac{1}{c}\sum_{n=1}^\infty nc_n\left( \frac{y}{2c} \right)^n = -\frac{1}{c}\sum_{n=1}^\infty nc_n \left\langle z_0^n, e^{yz_0/2c} \right\rangle \\
&= \left\langle \sum_{m=0}^\infty c_m z_0^m, -\frac{1}{c}\sum_{n=1}^\infty n \left(\frac{y}{2c} \right)^n \frac{z_0^n}{n!} \right\rangle = \left\langle g, \theta_{y/2c} \right\rangle.
\end{align*}

Suppose $f(z) = \sum_{r \geq 0} \sum_{p_r \in \mathcal{P}_r} c_{p_r}\frac{z^{p_r}}{\sqrt{p_r!}} \in B(\bC^4)$. Let $R > 2\sqrt{\sum_{i=0}^3(x^i)^2}/2c+1$, thus $\mathbf{x}/2cR \in B(0,1/2)$. Choose a $M \in \mathbb{N}$ large such that $rR^r < \sqrt{\lfloor r/4 \rfloor !} < \sqrt{p_r !}$ for $r \geq M$. Then we have
\begin{align*}
&\left| \frac{1}{c}\sum_{r=0}^\infty\sum_{p_r \in \mathcal{P}_r} c_{p_r}\frac{i}{(2c)^r}\frac{\mathbf{x}^{p_r}}{\sqrt{p_r!}} \right| \leq \frac{1}{c}\sum_{r=0}^\infty\sum_{p_r \in \mathcal{P}_r} |c_{p_r}|\frac{r}{(2c)^r}\frac{|\mathbf{x}^{p_r}|}{\sqrt{p_r!}}\\
&\leq \frac{R^M}{c}\sum_{r=0}^M \sum_{p_r \in \mathcal{P}_r} r|c_{p_r}| \frac{|\mathbf{x}^{p_r}|}{(2c)^r}\frac{1}{R^r}  + \frac{1}{c}\sum_{r=M+1}^\infty \sum_{p_r \in \mathcal{P}_r} r|c_{p_r}|\frac{|\mathbf{x}^{p_r}|}{(2c)^r}\frac{1}{R^r} \frac{R^r}{\sqrt{p_r!}} \\
&\leq \sup_{z \in B(0, 1/2)}\left[\frac{MR^M}{c}\sum_{r=0}^M \sum_{p_r \in \mathcal{P}_r}|c_{p_r}|\left|z^{p_r} \right| + \frac{1}{c}\sum_{n=M+1}^\infty \sum_{p_r \in \mathcal{P}_r}|c_{p_r}|\left|z^{p_r} \right| \right]\\
&= \frac{MR^M + 1}{c}|f|,
\end{align*}
whereby $|f|$ is the supremum norm defined on $B(\bC^4)$. This shows that we can identify $\theta_{x^0/2c}\chi_{x^1/2c}\chi_{x^2/2c}\chi_{x^3/2c}$ with a bounded linear functional $\left(\cdot, \theta_{x^0/2c}\chi_{x^1/2c}\chi_{x^2/2c}\chi_{x^3/2c}\right)_\sharp$. 
\end{proof}

\begin{cor}
For any $\mathbf{x} = (x^0, x^1, x^2, x^3) \in \bR^4$, define
\begin{align*}
\vartheta_{\mathbf{x}/2c}(z_0, z_1, z_2, z_3) &:= \theta_{x^0/2c}(z_0)\chi_{x^1/2c}(z_1)\chi_{x^2/2c}(z_2)\chi_{x^3/2c}(z_3) \\
&+\chi_{x^0/2c}(z_0)\theta_{x^1/2c}(z_1)\chi_{x^2/2c}(z_2)\chi_{x^3/2c}(z_3)\\
&+\chi_{x^0/2c}(z_0)\chi_{x^1/2c}(z_1)\theta_{x^2/2c}(z_2)\chi_{x^3/2c}(z_3) \\
&+\chi_{x^0/2c}(z_0)\chi_{x^1/2c}(z_1)\chi_{x^2/2c}(z_2)\theta_{x^3/2c}(z_3).
\end{align*}
Then $\left(\cdot, \vartheta_{\mathbf{x}/2c}\right)_\sharp$ is a bounded linear functional on $B(\bC^4)$, that maps \\ $f(z) \equiv f(z^0, z^1, z^2, z^3) \in B(\bC^4)$ to $\frac{d}{dc}f(\mathbf{x}/2c)$.
\end{cor}

\begin{proof}
A consequence of the chain rule.
\end{proof}

\begin{cor}\label{c.c.2}
Refer to Definition \ref{d.p.2}. For any $1 \leq \alpha \leq N$, we see that $c\nu_{R_\delta[a,T]}^{1/c,\alpha}$ is differentiable in $c$ and it defines a bounded linear functional $\left(\cdot, \beta_{R_\delta[a,T]}^{1/c,\alpha}\right)_\sharp$ on $\dB \otimes \mathfrak{g}$, given by Equation (\ref{e.b.4}). In other words,
\beq c \longmapsto \left(B, c\nu_{R_\delta[a,T]}^{1/c,\alpha} \right)_\sharp =  \int_{I_\delta^2}d\hat{s}\ \frac{1}{4c}\sum_{j=1}^3 |J_{0j}^{\sigma}|(\hat{s}) \left(B, \xi_{0j}(\sigma(\hat{s})/2c)\otimes E^\alpha\right)_\sharp , \label{e.d.2} \eeq is differentiable in $c$ and its derivative in $c$ is given by $\left(B, \beta_{R_\delta[a,T]}^{1/c,\alpha}\right)_\sharp$, $\beta_{R_\delta[a,T]}^{1/c,\alpha}$ defined in Equation (\ref{e.b.4}).
\end{cor}

\begin{proof}
Note that $\xi_{0j}(\sigma(\hat{s})/2c) = \frac{1}{\sqrt{2\pi}}e^{-|\sigma(\hat{s})|^2/8c^2}\chi_{\sigma(\hat{s})/2c}\otimes dx^0\wedge dx^j$, and its derivative is given by \beq \widehat{\beta_{0j}}(\sigma(\hat{s})/2c) := \frac{1}{\sqrt{2\pi}}e^{-|\sigma(\hat{s})|^2/8c^2}\left[\frac{|\sigma(\hat{s})|^2}{4c^3}\chi_{\sigma(\hat{s})/2c}
+ \vartheta_{\sigma(\hat{s})/2c}\right]\otimes dx^0 \wedge dx^j. \nonumber \eeq The derivative of $c\nu_{R_\delta[a,T]}^{1/c,\alpha}$, is hence given explicitly as
\begin{align}
\beta_{R_\delta[a,T]}^{1/c,\alpha}:=
\int_{I_\delta^2}d\hat{s}\ \frac{1}{4c}\sum_{j=1}^3 |J_{0j}^{\sigma}|(\hat{s}) \left[\widehat{\beta_{0j}}-\frac{1}{c}\xi_{0j}\right](\sigma(\hat{s})/2c)\otimes E^\alpha,
\label{e.b.4}
\end{align}
which defines a bounded linear functional on $\dB \otimes \mathfrak{g}$.
\end{proof}

\begin{lem}\label{l.s.1}
Refer to Item \ref{i.sb.3} in Remark \ref{r.rem.1}. Consider the Wiener space \\
$(|\cdot|, \fd_0 B_0(\bC^4), \tilde{\mu}_{\kappa^2})$, the norm $|\cdot|$ defined in Definition \ref{d.s.1}. We have for $q \geq 1$, $\bE[|\cdot|^q] < \infty$.
\end{lem}

\begin{proof}
For any $f \in \fd_0B_0(\bC^4)$, write $f(z) = \sum_{r \geq 0} \sum_{p_r \in \mathcal{P}_r} \tilde{c}_{p_r}\frac{z^{p_r}}{\sqrt{p_r!}}$. Let $(\cdot, c_{p_r})_\sharp: f\in \fd_0B_0(\bC^4) \mapsto \tilde{c}_{p_r} \in \bR$ be a real Gaussian random variable, mean 0 and variance $1/\kappa^2$. See Remark \ref{r.w.1}. Note that $\{(\cdot, c_{p_r})_\sharp:\ p_r \in \mathcal{P}_r,\ r \geq 0\}$ is a set of i.i.d. Gaussian random variables.

Now, there are at most $r^3$ elements in the set $\mathcal{P}_r$ for each $r \geq 0$. Let $M = \sum_{r=0}^\infty \frac{1}{2^r}|\mathcal{P}_r|$ and \beq c(q) = \frac{1}{\sqrt{2\pi}}\int_{-\infty}^\infty |x|^q e^{-x^2/2}\ dx. \nonumber \eeq

Apply Jensen's Inequality and Fubini's Theorem,
\begin{align*}
\bE\left[|\cdot|^q \right] \leq& \int_{f \in \fd_0B_0(\bC^4)} M^q\left[\frac{1}{M}\sum_{r \geq 0} \sum_{p_r \in \mathcal{P}_r} |(\cdot, c_{p_r})_\sharp | \left|\frac{1}{2^r} \right| \right]^q\ d\tilde{\mu}_{\kappa^2}(f) \\
\leq& M^{q-1}\sum_{r=0}^\infty \sum_{p_r \in \mathcal{P}_r}\left|\frac{1}{2^r} \right|\int_{f \in \fd_0B_0(\bC^4)}  |(\cdot, c_{p_r})_\sharp |^q \ d\tilde{\mu}_{\kappa^2}(f) \\
\leq& M^{q-1}\sum_{r=0}^\infty \left|\frac{1}{2^r} \right|\sum_{p_r \in \mathcal{P}_r} \frac{c(q)}{\kappa^q} \leq \frac{c(q)M^{q-1}}{\kappa^q}\sum_{n=0}^\infty \left|\frac{r^3}{2^r} \right|< \infty.
\end{align*}
\end{proof}

\begin{cor}\label{c.s.1}
For any $1 \leq \alpha \leq N$, write
\begin{align*}
B_\alpha &= \left[\sum_{i=1}^3 B_i \otimes dx^0 \wedge dx^i + \sum_{1 \leq i < j \leq 3}B_{ij}\otimes dx^i \wedge dx^j \right] \otimes E^\alpha \\
&\in \dB \otimes \mathfrak{g} \subset B(\bC^4) \otimes \Lambda^2(\bR^4) \otimes \mathfrak{g},
\end{align*}
and \beq |B_\alpha| = \sum_{i=1}^3 |B_i|, \label{e.s.2} \eeq the norm $|\cdot|$ is the supremum norm on $B(\bC^4)$, defined in Definition \ref{d.s.1}. There exists a constant $C(c_0,R_\delta[a,T]) >0$ such that for $c > c_0$, we have
\begin{align}
&\left|\left(B_\alpha, c\nu_{R[a,T]}^{1/c,\alpha} \right)_\sharp \left(B_\alpha, \beta_{R_\delta[a,T]}^{1/c,\alpha} \right)_\sharp \right|+ \left|\left(B_\alpha, \beta_{R[a,T]}^{1/c,\alpha} \right)_\sharp \left(B_\alpha, c\nu_{R_\delta[a,T]}^{1/c,\alpha} \right)_\sharp \right| \nonumber\\
&\leq C(c_0, R_\delta[a,T])^2 |B_\alpha|^2,\label{e.d.3}
\end{align}
and $B_\alpha \in \dB \otimes \mathfrak{g} \mapsto |B_\alpha|^2$ is $L^q$ integrable with respect to Wiener measure $\tilde{\mu}_{\kappa^2}^{\times^{N}}$, for any $q \geq 1$.
\end{cor}

\begin{proof}
Let $\mathbf{x} = (x^0, x^1, x^2, x^3)$. We saw in the proof of Proposition \ref{p.s.1}, that the functional norm
$\parallel (\cdot, \theta_{x^0/2c}\chi_{x^1/2c}\chi_{x^2/2c}\chi_{x^3/2c})_\sharp \parallel$ depends on $|\mathbf{x}/2c|$ and $c$. Likewise, one can also show that the functional norm $\parallel (\cdot, \chi_{\mathbf{x}/2c})_\sharp \parallel$ depends on $|\mathbf{x}/2c|$. Since the flat surface $\frac{1}{2c}R_\delta[a, T]$ is a bounded rectangular surface and $R[a,T]/2c \subset R_\delta[a,T]/2c_0$ for any $c > c_0 > 0$, we can find a constant $C(c_0, R_\delta[a,T])$ such that the functional norm \beq \left\Vert\left(\cdot, c\nu_{R_\delta[a,T]}^{1/c,\alpha} \right)_\sharp \right\Vert + \left\Vert\left(\cdot, \beta_{R_\delta[a,T]}^{1/c,\alpha} \right)_\sharp \right\Vert < C(c_0, R_\delta[a,T]), \nonumber \eeq and thus the above Inequality (\ref{e.d.3}) follows from Equations (\ref{e.d.2}) and (\ref{e.b.4}). The $L^q$ integrability follows from the previous lemma.
\end{proof}

\end{document}